\documentclass[11pt,letterpaper]{article}
\usepackage[margin=1in]{geometry}
\usepackage{amsmath,amssymb,color,ifpdf,latexsym,graphicx,url}
\usepackage{enumerate}
\usepackage{subcaption}

\usepackage{array}
\newcolumntype{L}[1]{>{\raggedright\let\newline\\\arraybackslash\hspace{0pt}}m{#1}}
\newcolumntype{C}[1]{>{\centering\let\newline\\\arraybackslash\hspace{0pt}}m{#1}}
\usepackage[table]{xcolor}

\let\realbfseries=\bfseries
\def\bfseries{\realbfseries\boldmath}

\let\realbibitem=\bibitem
\def\bibitem{\par \vspace{-1.2ex}\realbibitem}

\let\epsilon=\varepsilon
\def\defn#1{\textbf{\textit{\boldmath #1}}}

\usepackage{hyperref}
\usepackage{amsthm, amsmath}
\theoremstyle{plain}
\newtheorem{theorem}{Theorem}[section]
\newtheorem{lemma}[theorem]{Lemma}

\theoremstyle{definition}

\theoremstyle{remark}

{\makeatletter
 \gdef\xxxmark{%
   \expandafter\ifx\csname @mpargs\endcsname\relax 
     \expandafter\ifx\csname @captype\endcsname\relax 
       \marginpar{xxx}
     \else
       xxx 
     \fi
   \else
     xxx 
   \fi}
 \gdef\xxx{\@ifnextchar[\xxx@lab\xxx@nolab}
 \long\gdef\xxx@lab[#1]#2{\textbf{[\xxxmark #2 ---{\sc #1}]}}
 \long\gdef\xxx@nolab#1{\textbf{[\xxxmark #1]}}
 \long\gdef\xxx@lab[#1]#2{}\long\gdef\xxx@nolab#1{}%
}

\numberwithin{equation}{section}

\title{This Game Is Not Going To Analyze Itself}

\author{
  Aviv Adler%
    \thanks{Computer Science and Artificial Intelligence Laboratory,
      Massachusetts Institute of Technology, Cambridge, MA 02139, USA,
      \protect\url{{adlera,joshuaa,lkdc,mcoulomb,edemaine,diomidova,dylanhen,jaysonl}@mit.edu}}
\and
  Hayashi  Ani\footnotemark[1]
\and
  Lily Chung\footnotemark[1]
\and
  Michael Coulombe\footnotemark[1]
\and
  Erik D. Demaine\footnotemark[1]
\and
  Jenny Diomidova\footnotemark[1]
\and
  Dylan Hendrickson\footnotemark[1]
\and
  Jayson Lynch\footnotemark[1]
}

\date{}

\begin{document}
\maketitle

\begin{abstract}
  We analyze the puzzle video game \emph{This Game Is Not Going To Load Itself},
  where the player routes data packets of three different colors
  from given sources to given sinks of the correct color.
  Given the sources, sinks, and some previously placed arrow tiles,
  we prove that the game is in $\Sigma_2^P$; in NP for sources of equal period;
  NP-complete for three colors and six equal-period sources with player input;
  and even without player input, simulating the game is both NP- and coNP-hard
  for two colors and many sources with different periods.
  On the other hand, we characterize which locations for three data sinks
  admit a \emph{perfect} placement of arrow tiles that guarantee correct routing
  no matter the placement of the data sources,
  effectively solving most instances of the game as it is normally played.
\end{abstract}



\def\gamename{TGINGTLI}

\section{Introduction}
\label{sec:intro}

\emph{This Game Is Not Going To Load Itself} (\gamename{}) \cite{tgingtli} is a free game created in 2015 by Roger ``atiaxi'' Ostrander for the Loading Screen Jam, a game jam hosted on \protect\url{itch.io},
where it finished $7$th overall out of $46$ entries.
This game jam was a celebration of the expiration of US Patent 5,718,632 \cite{US5718632}, which covered the act of including mini-games during video game loading screens.
In this spirit, \gamename{} is a real-time puzzle game themed around the player helping a game load three different resources of itself --- save data, gameplay, and music, colored red, green, and blue --- by placing arrows on the grid cells to route data entering the grid to a corresponding sink cell.
Figure~\ref{fig:realplay} shows an example play-through.

\begin{figure}
  \centering
  \includegraphics[scale=0.63]{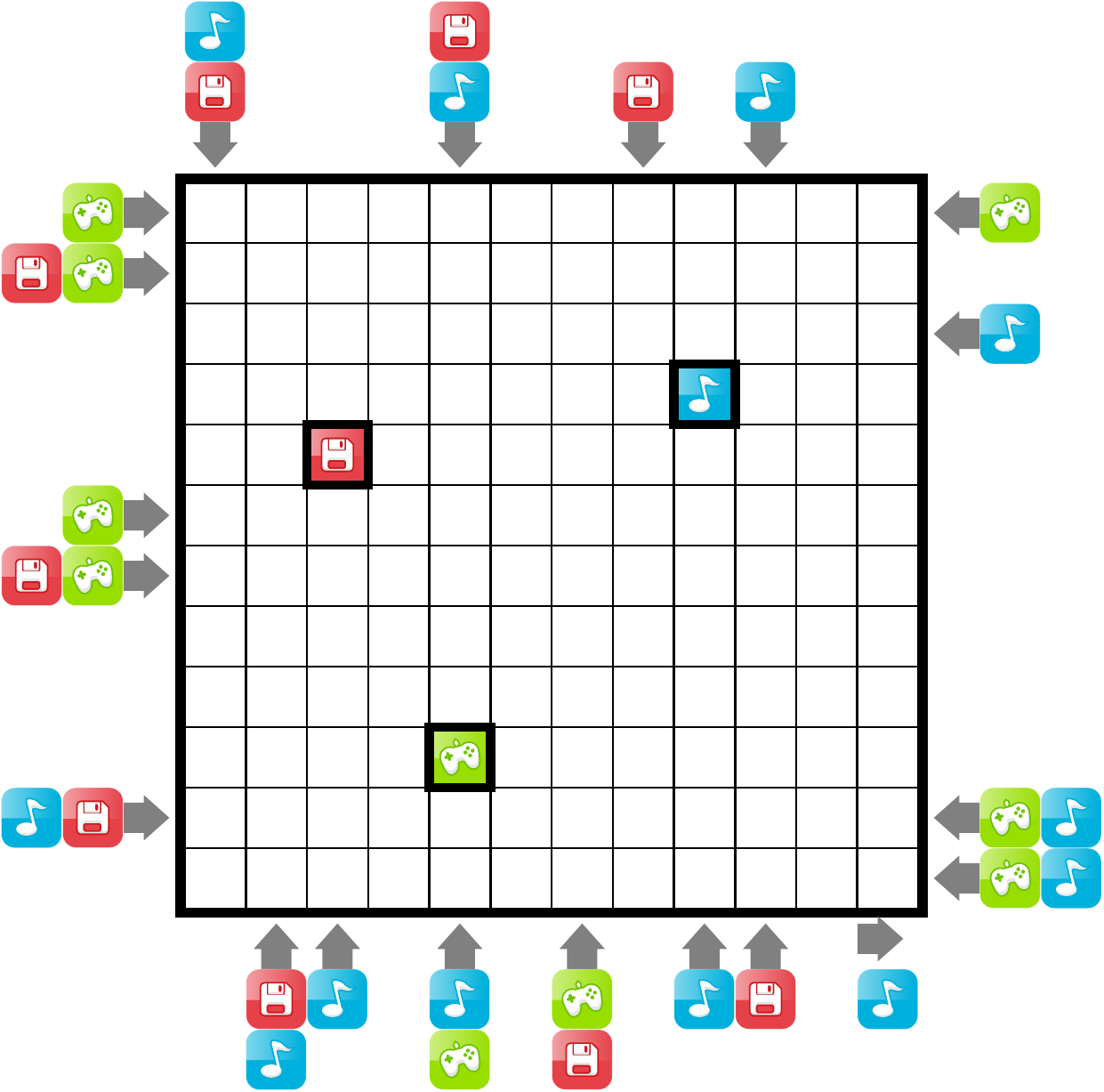}\hfill
  \includegraphics[scale=0.63]{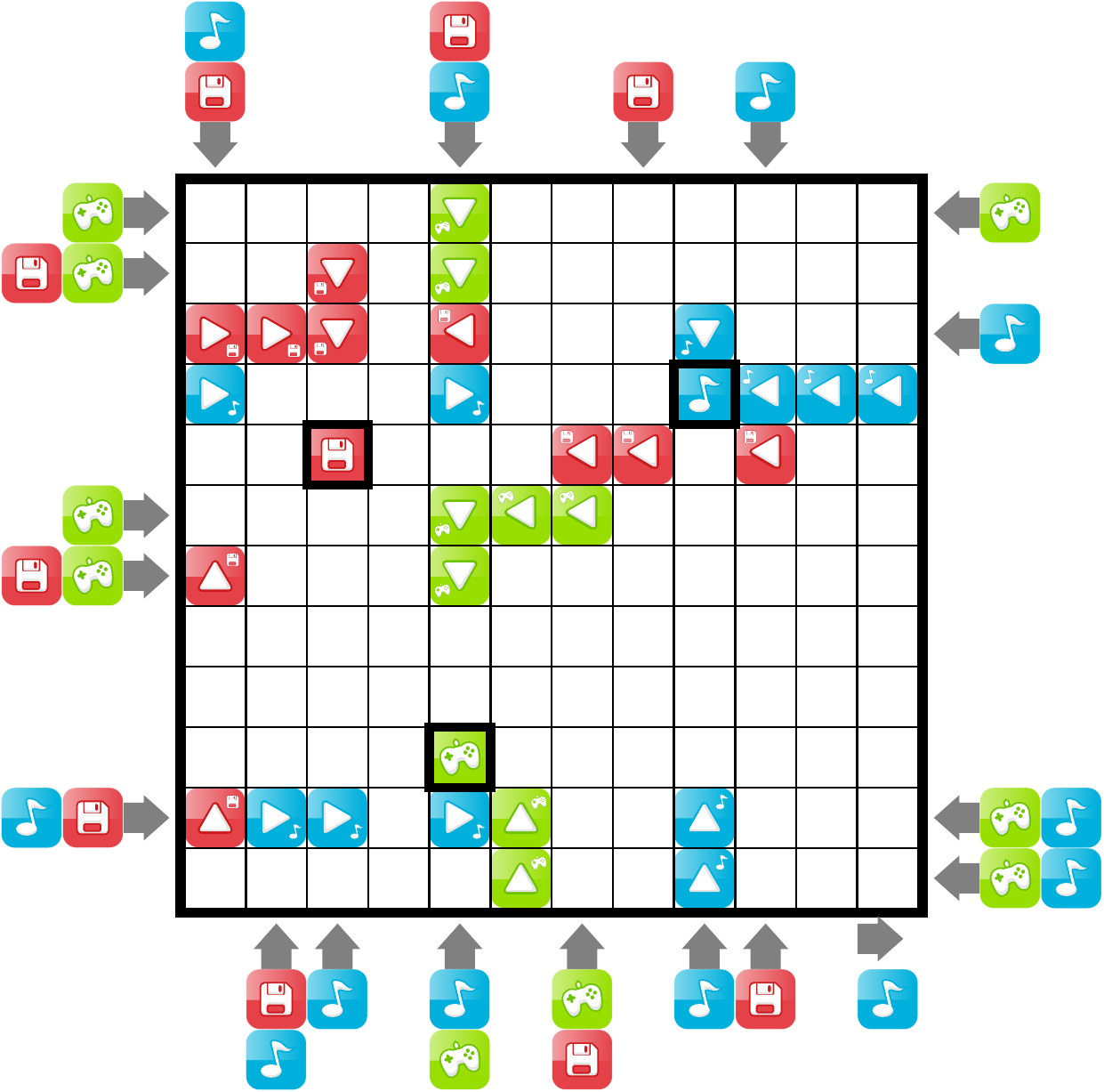}
  \caption{%
    Left: The (eventual) input for a real-world Level 16 in \gamename{}.
    Right: A successful play-through that
    routes every packet to its corresponding sink.}
  \label{fig:realplay}
\end{figure}

We formalize \gamename{} as follows.
You are given an $m \times n$ grid where each unit-square cell is either empty,
contains a data sink, or contains an arrow pointing
in one of the four cardinal directions.
(In the implemented game, $m=n=12$ and no arrows are placed initially.)
Each data sink and arrow has a color (resource) of red, green, or blue;
and there is exactly one data sink of each color in the grid.
In the online version (as implemented), sources appear throughout the game;
in the offline version considered here, all sources are known a priori.
Note that an outer edge of the grid may have multiple sources of different colors.
Finally, there is an loading bar that starts at an integer $k_0$
and has a goal integer~$k^*$.

During the game,
each source periodically produces data packets of its color,
which travel at a constant speed into the grid.
If a packet enters the cell of an arrow of the same color,
then the packet will turn in that direction.
(Arrows of other colors are ignored.)
If a packet reaches the sink of its color,
then the packet disappears and the loading bar increases by one unit of data.
If a packet reaches a sink of the wrong color, or exits the grid entirely,
then the packet disappears and the loading bar decreases by one unit of data, referred to as taking damage.
Packets may also remain in the grid indefinitely
by going around a cycle of arrows;
this does not increase or decrease the loading bar.
The player may at any time permanently fill an empty cell with an arrow, which may be of any color and pointing in any of the four directions.
If the loading bar hits the target amount $k^*$, then the player wins; but if the loading bar goes below zero, then the player loses.

In Section~\ref{sec:NP-hardness},
we prove NP-hardness of the {\gamename} decision problem:
given a description of the grid
(including sources, sinks, and preplaced arrows),
can the player place arrows to win?
This reduction works even for just six sources and three colors;
it introduces a new problem, \defn{3DSAT}, where variables have three
different colors and each clause mixes variables of all three colors.
In Section~\ref{sec:sigma-2}, we introduce more detailed models for the
periodic behavior of sources, and show that many sources of differing periods
enable both NP- and coNP-hardness of winning the game,
\emph{even without player input} (just simulating the game).
On the positive side, we prove that these problems are in $\Sigma_2^P$;
and in NP when the source periods are all equal,
as in our first NP-hardness proof, so this case is in fact NP-complete.

In Section~\ref{sec:perfect-layouts}, we consider how levels start in
the implemented game: a grid with placed sinks but no preplaced arrows.
We give a full characterization of when there is a \defn{perfect layout}
of arrows, where all packets are guaranteed to route to the correct sink,
\emph{no matter where sources get placed}.
In particular, this result provides a winning strategy for most
sink arrangements in the implemented game.
Notably, because this solution works independent of the sources,
it works in the online setting.

\xxx{We should perhaps point to some of the other hardness of videogames literature and if there is literature on solving small instances of puzzle games through math and exhaustive search.}

\section{NP-Hardness for Three Colors and Six Sources}
\label{sec:NP-hardness}

We first prove that \gamename{} is NP-hard by reducing from a new problem called
\defn{3-Dimensional SAT (3DSAT)},
defined by analogy to 3-Dimensional Matching (3DM).
3DSAT is a variation of 3SAT where, in addition to a 3CNF formula,
the input specifies one of three colors (red, green, or blue) to each variable
of the CNF formula, and the CNF formula is constrained to have trichromatic
clauses, i.e., to have exactly one variable (possibly negated) of each color.

\begin{lemma}
  3DSAT is NP-complete.
\end{lemma}

\begin{proof}
  We reduce from 3SAT to 3DSAT by converting a 3CNF formula $F$
  into a 3D CNF formula~$F'$.
  For each variable $x$ of~$F$,
  we create three variables $x^{(1)}, x^{(2)}, x^{(3)}$ in $F'$
  (intended to be equal copies of $x$ of the three different colors)
  and add six clauses to $F'$ to force $x^{(1)} = x^{(2)} = x^{(3)}$:
  \def\halfup#1{\raisebox{2.5ex}{\smash{$#1$}}}
  \begin{align*}
    \lnot x^{(1)} \lor x^{(2)} \lor x^{(3)}       &\iff (x^{(1)} \to x^{(2)}) \lor x^{(3)} \\
    \lnot x^{(1)} \lor x^{(2)} \lor \lnot x^{(3)} &\iff (x^{(1)} \to x^{(2)}) \lor \lnot x^{(3)} ~\halfup{\Bigg\rbrace \iff x^{(1)} \to x^{(2)}} \\
    x^{(1)} \lor \lnot x^{(2)} \lor x^{(3)}       &\iff (x^{(2)} \to x^{(3)}) \lor x^{(1)} \\
    \lnot x^{(1)} \lor \lnot x^{(2)} \lor x^{(3)} &\iff (x^{(2)} \to x^{(3)}) \lor \lnot x^{(1)} ~\halfup{\Bigg\rbrace \iff x^{(2)} \to x^{(3)}} \\
    x^{(1)} \lor x^{(2)} \lor \lnot x^{(3)}       &\iff (x^{(3)} \to x^{(1)}) \lor x^{(2)}\\
    x^{(1)} \lor \lnot x^{(2)} \lor \lnot x^{(3)} &\iff (x^{(3)} \to x^{(1)}) \lor \lnot x^{(2)} ~\halfup{\Bigg\rbrace \iff x^{(3)} \to x^{(1)}}
  \end{align*}
  Thus the clauses on the left are equivalent to the implication loop
  $x^{(1)} \implies x^{(2)} \implies x^{(3)} \implies x^{(1)}$,
  which is equivalent to $x^{(1)} = x^{(2)} = x^{(3)}$.
  
  For each clause $c$ of $F$ using variables $x$ in the first literal,
  $y$ in the second literal, and $z$ in the third literal,
  we create a corresponding clause $c'$ in $F'$ using
  $x^{(1)}$, $y^{(2)}$, and $z^{(3)}$ (with the same negations as in~$c$).
  All clauses in $F'$ (including the variable duplication clauses above)
  thus use a variable of the form $x^{(i)}$ in the
  $i$th literal for $i \in \{1,2,3\}$,
  so we can 3-color the variables accordingly.
\end{proof}

\begin{theorem} \label{thm:NP-hard36}
  {\gamename} is NP-hard, even with three colors and six sources
  of equal period.
\end{theorem}

\begin{proof}
  Our reduction is from 3DSAT.
  Figure~\ref{fig:reduction-sketch} gives a high-level sketch:
  each variable gadget has two possible routes
  for a packet stream of the corresponding color,
  and each clause gadget allows at most two colors of packets
  to successfully route through.
  When a clause is satisfied by at least one variable,
  the clause gadget allows the other variables to successfully pass
  through; otherwise, at least one of the packet streams enters a cycle.
  Variables of the same color are chained together
  to re-use the packet stream of that color.

  \begin{figure}
    \centering
    \subcaptionbox{Satisfied clause}{\includegraphics[scale=0.5]{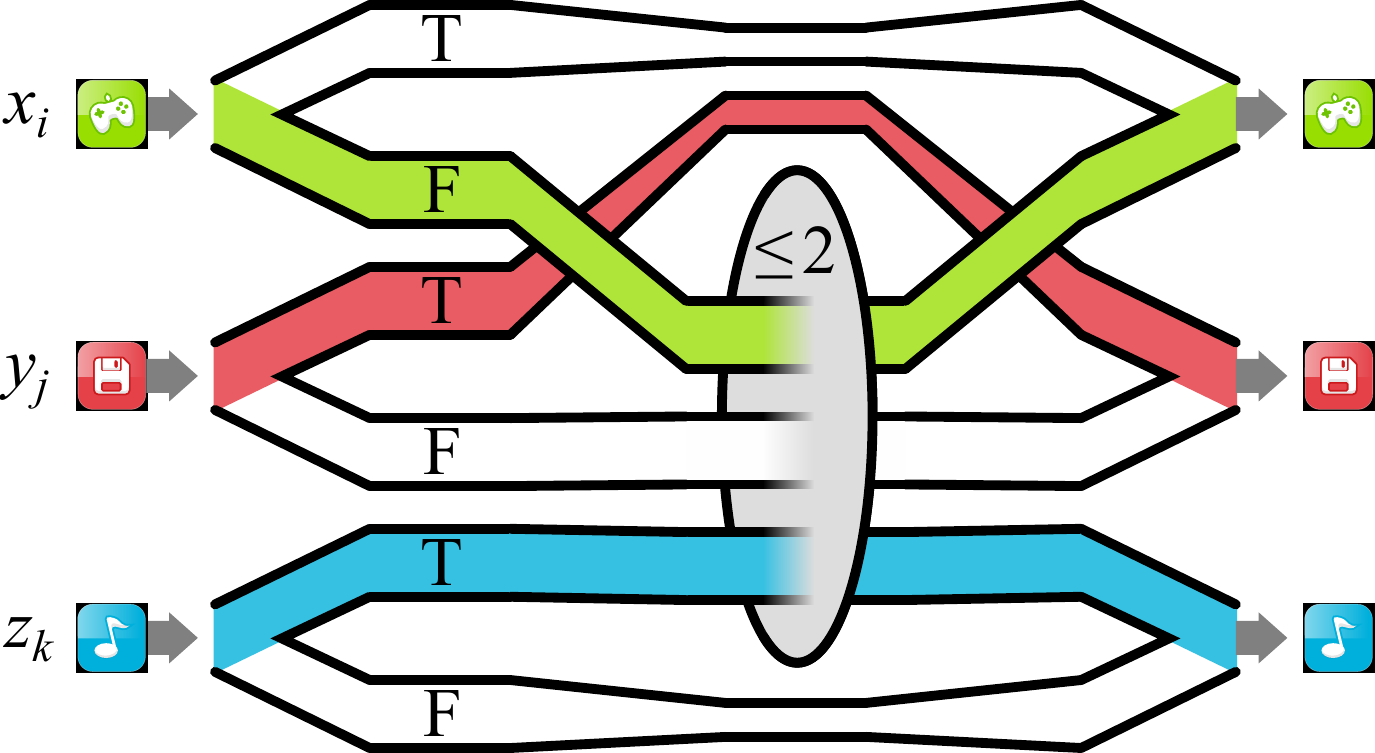}}\hfil
    \subcaptionbox{Unsatisfied clause}{\includegraphics[scale=0.5]{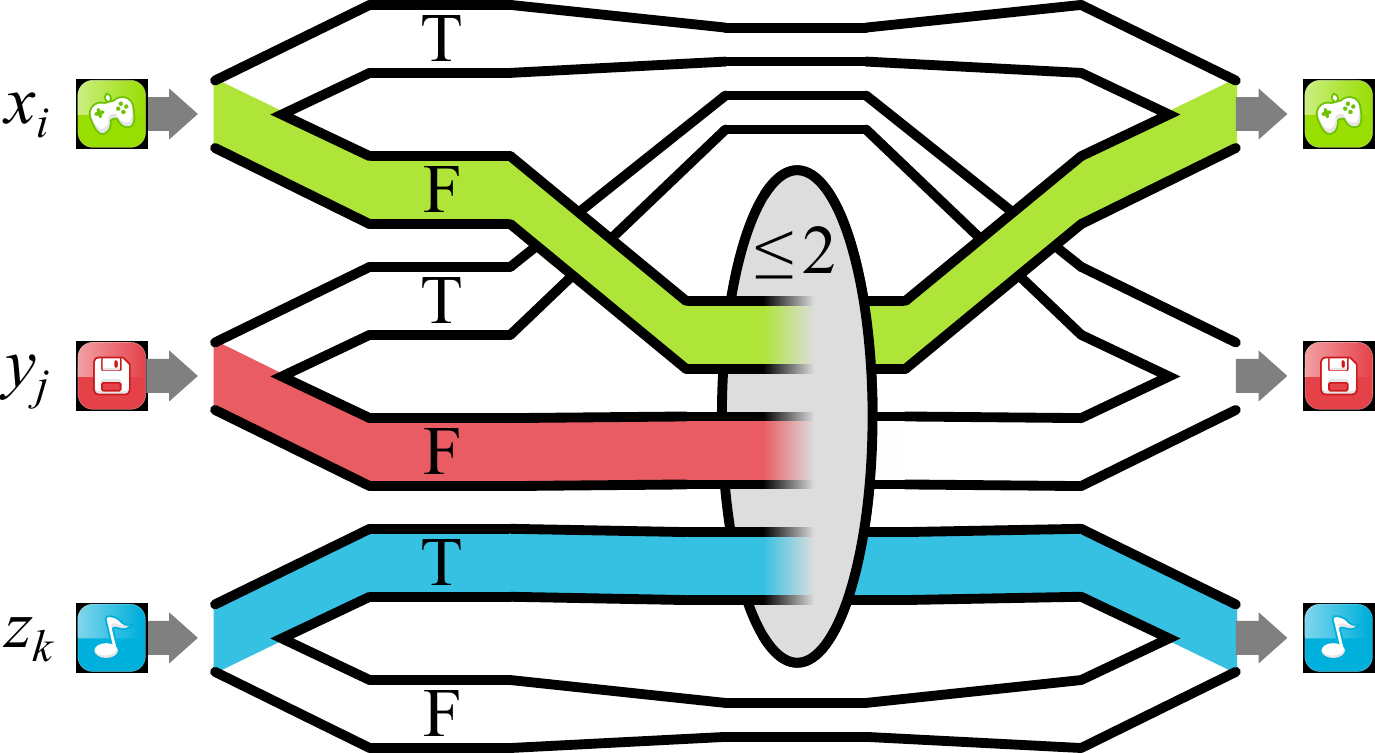}}
    \caption{Sketch of our NP-hardness reduction.}
    \label{fig:reduction-sketch}
  \end{figure}
  
  In detail, most cells of the game board
  will be prefilled, leaving only a few empty cells
  (denoted by question marks in our figures) that the player can fill.

  For each color, say red, we place a red source gadget
  on the left edge of the construction.
  Then, for each red variable $x$ in sequence,
  we place a variable gadget of Figure~\ref{fig:gadget:variable}
  at the end of the red stream.
  To prevent the packets from entering a loop,
  the player must choose between sending the stream upward or downward,
  which results in it following one of the two rightward paths,
  representing the literals $x$ and $\overline x$ respectively.
  The path followed by the packet stream is viewed as \emph{false};
  an empty path is viewed as \emph{true}.

\begin{figure}[t]
  \centering
  \begin{minipage}[b]{0.51\linewidth}
    \centering
    \includegraphics[scale=0.9]{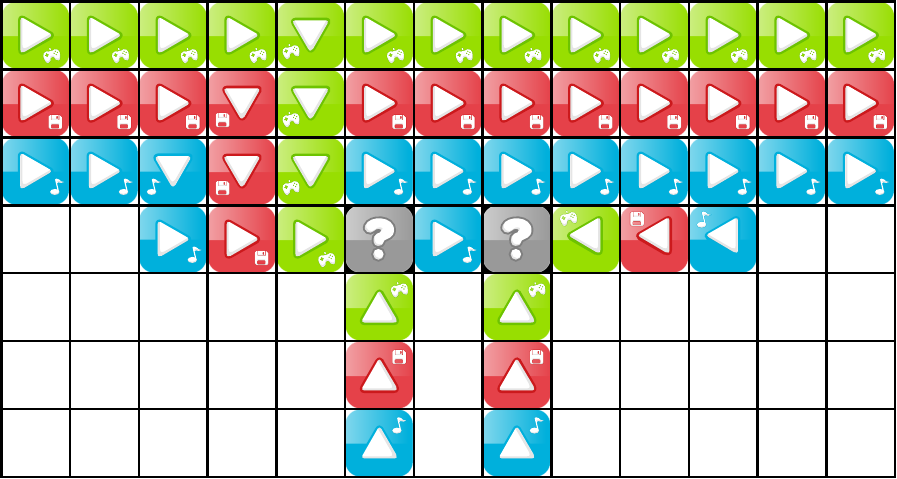}
    \caption{Clause gadget. At most two streams of data, representing false literals, can pass through the gadget (by placing upwards arrows in the ``?'' cells) without entering a cycle.
      Placing any other direction of arrow also puts the stream into a cycle.}
    \label{fig:gadget:clause}
  \end{minipage}\hfill
  \begin{minipage}[b]{0.22\linewidth}
    \centering
    \includegraphics[scale=0.9]{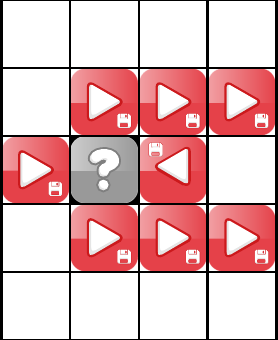}
    \caption{\centering Variable gadget lets the player route packets from a source to one of two literal paths.}
    \label{fig:gadget:variable}
  \end{minipage}\hfill
  \begin{minipage}[b]{0.22\linewidth}
    \centering
    \includegraphics[scale=0.9]{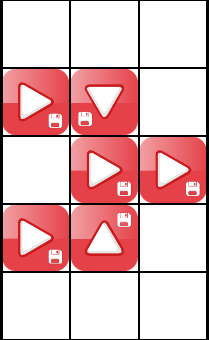}
    \caption{\centering Merge gadget combines two literal paths back into one.\linebreak~}
    \label{fig:gadget:merge}
  \end{minipage}
  
  \bigskip
  
  \begin{minipage}[b]{0.5\linewidth}
    \centering
    \includegraphics[scale=0.9]{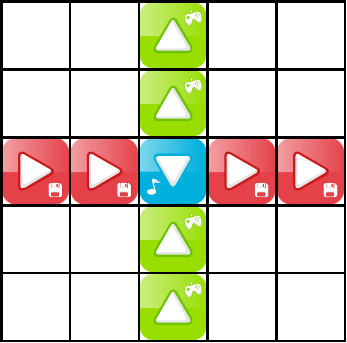}
    \caption{Crossover gadget between two literal paths of same or different colors. (The center cell is colored different from both paths.)}
    \label{fig:gadget:crossover}
  \end{minipage}\hfill
  \begin{minipage}[b]{0.46\linewidth}
    \includegraphics[scale=0.9]{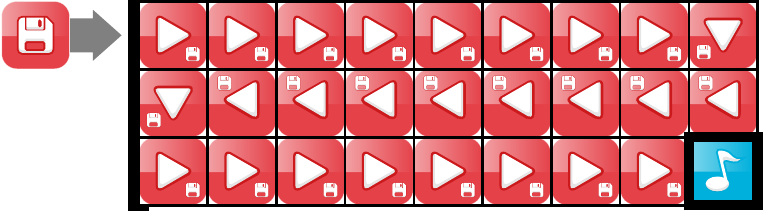}
    \caption{Damage gadget forces damage at a unit rate after a desired start delay.\hspace{\fill}\linebreak~}
    \label{fig:gadget:damage}
  \end{minipage}
\end{figure}

  Then we route each literal path to sequentially visit
  every clause containing it.
  Figure~\ref{fig:gadget:crossover} shows a crossover gadget
  to enable such routing.
  (Note that it still works if both lines are the same color.)

  Figure~\ref{fig:gadget:clause} shows the clause gadget,
  which allows at most two packet streams to pass through it.
  If all three literals are false then at least
  one stream of data must be placed into a cycle.
  On the other hand, if the clause is satisfied, then the literal paths
  carrying data can pass their data on to the next clause, and so on.
  Note that the length of the diverted path through the clause
  is the same for all three colors.

  After all uses of the red variable \(x\),
  we lengthen the two red literal paths
  corresponding to \(x\) and \(\overline{x}\) to have the same length,
  then combine them back together
  using the merge gadget of Figure~\ref{fig:gadget:merge}.
  We then route this red path into the variable gadget
  for the next red variable, and so on.
  Finally, after all red variables, we connect the red stream
  to a red sink.

  We lengthen the red, green, and blue streams
  to all have the same length~$\ell$.
  If the player successfully satisfies all clauses, then
  they will increase the loading bar by $3$ units ($1$~per color)
  after an initial delay of~$\ell$.
  We set the parameters so that the player wins in this case:
  $k^* - k_0 = 3$.
  Otherwise, the loading rate is at most~$2$.
  To ensure that the player loses in this case, we add
  $3$ damage gadgets of Figure~\ref{fig:gadget:damage},
  each incurring damage at a rate of $1$ after an initial delay of $\ell+1$.
  Thus we obtain a net of $-1$ per period,
  so the player eventually loses even if $k_0$ is large.
\end{proof}


This NP-hardness result does not need a very specific model of sources and
how they emit packets.
To understand whether the problem is in NP, we need a more specific model,
which is addressed in the next section.


\section{Membership in $\Sigma_2^P$ and Hardness from Source Periodicity}
\label{sec:sigma-2}

In this section, we consider the effect of potentially differing periods
for different sources emitting packets.
Specifically, we show that carefully setting periods together with
the unbounded length of the game results in both NP- and coNP-hardness
of determining the outcome of {\gamename},
\emph{even when the player is not making moves}.
Conversely, we prove that the problem is in~$\Sigma_2^P$,
even allowing player input.

\subsection{Model and Problems}

More precisely, we model each source $s$ as emitting data packets of its color
into the grid with its own period $p_s$, after a small warmup time $w_s$
during which the source may emit a more specific pattern of packets.
{\gamename} as implemented has a warmup behavior of each source
initially (upon creation) waiting 5 seconds before the first emitted packet,
then emitting a packet after 2 seconds, after 1.999 seconds,
after 1.998 seconds, and so on,
until reaching a fixed period of 0.5 seconds.
This is technically a warmup period of 1881.25 seconds with
1500 emitted packets, followed by a period of 0.5 seconds.

In the \defn{simulation} problem, we are given the initial state of the grid,
a list of timestamped \defn{events} for
when each source emits a packet during its warmup period,
when each source starts periodic behavior,
and when the player will place each arrow.
We assume that timestamps are encoded in binary but
(to keep things relatively simple)
periods and warmup times are encoded in unary.
The problem then asks to predict whether the player wins;
that is, the loading bar reaches~$k^*$ before a loss occurs.

In the \defn{game} problem, we are not given the player's arrow placements.
If we allow nondeterministic algorithms, the game problem reduces to the
simulation problem: just guess what arrows we place, where, and at what times.

A natural approach to solving the simulation problem is to simulate the game
from the initial state to each successive event.
Specifically, given a state of the game (a grid with sinks, sources of varying periods and offsets, placed arrows, and the number of in-flight packets at each location) and a future timestamp \(t\),
we wish to determine the state of the game at time \(t\).
Using this computation, we can compute future states of the game quickly by
``skipping ahead'' over the time between events.
On an $m \times n$ grid, there are $O(m n)$ events, so we can determine the state of the game at any time \(t\) by simulating polynomially many phases between events.

This computation is easy to do.
Given the time \(t\), we can divide by each source's period and the period of each cycle of arrows to determine
how many packets each source produces and where the arrows route them ---
either to a sink which affects loading, off the grid, stuck in a cycle,
or in-flight outside a cycle ---
and then sum up the effects to obtain the new amount loaded and the number of packets at each location.
\xxx{Be more specific?}

However, being able to compute future states
does not suffice to solve the simulation and game problems because there might be an intermediate time
where the loading amount drops below $0$ or reaches the target amount~$k^*$.
Nonetheless, this suffices to show that the problems are in \(\Sigma_2^P\),
by guessing a win time and verifying there are no earlier loss times:

\begin{lemma}
  The simulation and game problems are in \(\Sigma_2^P\).
\end{lemma}
\begin{proof}
  The player wins if there exists a time with a win such that all smaller times are not losses.
  To solve the simulation problem,
  nondeterministically guess the winning time and verify that it is a win
  by computing the state at that time.
  Then check using a coNP query that there was no loss before that time,
  again using the ability to quickly compute states at individual timestamps.

  To solve the game problem, we first existentially guess the details of the arrow placements,
  then solve the resulting simulation problem as before.
\end{proof}

An easier case is when the source periods are all the same
after warmup, as implemented in the real game.
Theorem~\ref{thm:NP-hard36} proved this version of the game NP-hard,
and we can now show that it is NP-complete:

\begin{lemma}
  If all sources have the same polynomial-length period after a
  polynomial number $t_p$ of time steps,
  then the simulation problem is in P and the game problem is in NP.
\end{lemma}

\begin{proof}
  In this case, we can check for wins or losses in each phase between events by explicitly simulating longer than all packet paths,
  at which point the loading bar value becomes periodic with the common source period.
  (Cycles of arrows may have different periods but these do not affect the loading bar, and thus do not matter when checking for wins and losses.)
  We skip over each phase, checking for win or loss along the way.
  If the game continues past the last event, we measure the sign of the net score change over the period.
  If it is positive, the player will eventually win;
  if it is negative, the player will eventually lose; and
  if it is zero, the game will go on forever.
\end{proof}


In the remainder of this section, we consider the case where each source can be assigned any integer period, and the period does not change over time.

\subsection{Periodic Sum Threshold Problem}

With varying source periods, the challenge is that the overall periodic
behavior of the game can have an extremely large (exponential) period.

We can model this difficulty via the
\defn{Periodic Sum Threshold Problem}, defined as follows.
We are given a function $f(x) = \sum_{i} g_i(x)$
where each $g_i$ has unary integer period $T_i$
and unary maximum absolute value $M_i$.
In addition, we are given a unary integer $\tau > 0$
and a binary integer time $x^*$.
The goal is to determine whether there exists an integer $x$
in $[0, x^*)$ such that $f(x) \geq \tau$.
(Intuitively, reaching $\tau$ corresponds to winning.)

\begin{theorem}
\label{thm:pstp-np-complete}
  The Periodic Sum Threshold Problem is NP-complete,
  even under the following restrictions:
  \begin{enumerate}
  \item \label{prop:one-hot}
    Each $|g_i|$ is a one-hot function, i.e.,
    $g_i(x) = 0$ everywhere except for exactly one $x$ in its period
    where $g_i(x) = \pm 1$.
  \item \label{prop:lambda}
    We are given a unary integer $\lambda < 0$ such that
    $f(x) > \lambda$ for all $0 \leq x < x^*$ and $f(x^*) \leq \lambda$.
    (Intuitively, dipping down to $\lambda$ corresponds to losing.)
  \end{enumerate}
\end{theorem}
\begin{proof}
  First, the problem is in NP: we can guess $x \in [0,x^*)$
  and then evaluate whether $f(x) \leq c$ in polynomial time.

  For NP-hardness, we reduce from 3SAT.
  We map each variable $v_i$ to the $i$th prime number $p_i$ excluding~$2$.
  Using the Chinese Remainder Theorem,
  we can represent a Boolean assignment $\phi$ as a single integer $0 \le x < \prod_i p_i$
  where $x \equiv 1 \mod p_i$ when $\phi$ sets $v_i$ to true,
  and $x \equiv 0 \mod p_i$ when $\phi$ sets $v_i$ to false.
  (This mapping does not use other values of $x$ modulo~$p_i$.
  In particular, it leaves $x \equiv -1 \mod p_i$ unused,
  because $p_i \geq 3$.)

  Next we map each clause such as $C = (v_i \vee v_j \vee \neg v_k)$
  to the function
  $$g_C(x) = \max\{[x \equiv 1 \mod p_i], [x \equiv 1 \mod p_j], [x \equiv 0 \mod p_k]\},$$
  i.e., positive literals check for $x \equiv 1$
  and negated literals check for $x \equiv 0$.
  This function is $1$ exactly when $x$ corresponds
  to a Boolean assignment that satisfies~$C$.
  This function has period $p_i p_j p_k$,
  whose unary value is bounded by a polynomial.
  Setting $\tau$ to the number of clauses,
  there is a value $x$ where the sum is $\tau$ if and only if
  there is a satisfying assignment for the 3SAT formula.
  (Setting $\tau$ smaller, we could reduce from Max 3SAT.)

  To achieve Property~\ref{prop:one-hot}, we split each $g_C$ function
  into a sum of polynomially many one-hot functions
  (bounded by the period).  In fact, seven functions per clause suffice,
  one for each satisfying assignment of the clause.

  To achieve Property~\ref{prop:lambda},
  for each prime $p_i$, we add the function
  $h_i(x) = -[x \equiv -1 \mod p_i]$.
  This function is $-1$ only for unused values of \(x\)
  which do not correspond to any assignment \(\phi\),
  so it does not affect the argument above.
  Setting $-\lambda$ to the number of primes (variables)
  and $x^* = \prod_i p_i - 1$,
  we have $\sum_i h_i(x^*) = \lambda$
  because $h_i(x^*) \equiv -1 \mod p_i$ for all~$i$,
  while $\sum_i h_i(x) > \lambda$ for all $0 \leq x < x^*$.
  All used values \(x\) are smaller than \(x^*\).

  In total, $f(x)$ is the sum of the constructed functions
  and we obtain the desired properties.
\end{proof}


\subsection{Simulation Hardness for Two Colors}

We can use our hardness of the Periodic Sum Threshold Problem
to prove hardness of simulating {\gamename}, even without player input.

\begin{theorem} \label{thm:sim NP-hard}
  Simulating {\gamename} and determining whether the player wins
  is NP-hard, even with just two colors.
\end{theorem}
\begin{proof}
  We reduce from the Periodic Sum Threshold Problem proved NP-complete
  by Theorem~\ref{thm:pstp-np-complete}.

  For each function $g_i$ with one-hot value $g_i(x_i) = 1$ and period~$T_i$,
  we create a blue source $b_i$ and a red sources $r_i$,
  of the same emitting period~$T_i$,
  and route red and blue packets from these sources to the blue sink.
  By adjusting the path lengths and/or the warmup times of the sources,
  we arrange for a red packet to arrive one time unit after each blue packet
  which happens at times $\equiv x_i \mod T_i$.
  Thus the net effect on the loading bar value is $+1$ at time $x_i$
  but returns to $0$ at time $x_i + 1$.
  Similarly, for each function $g_i$ with one-hot value $g_i(x_i) = -1$,
  we perform the same construction but swapping the roles of red and blue.

  Setting $k_0 = -\lambda-1 \geq 0$, the loading bar goes negative
  (and the player loses)
  exactly when the sum of the functions $g_i$ goes down to~$\lambda$.
  Setting $k^* = k_0 + \tau$, the loading bar reaches $k^*$
  (and the player wins)
  exactly when the sum of the functions $g_i$ goes up to~$\tau$.
\end{proof}

This NP-hardness proof relies on completely different aspects of the game
from the proof in Section~\ref{sec:NP-hardness}: instead of using player input,
it relies on varying (but small in unary) periods for different sources.
More interesting is that we can also prove the same problem coNP-hard:

\begin{theorem}
  Simulating {\gamename} and determining whether the player wins
  is coNP-hard, even with just two colors.
\end{theorem}
\begin{proof}
  We reduce from the complement of the Periodic Sum Threshold Problem,
  which is coNP-complete by Theorem~\ref{thm:pstp-np-complete}.
  The goal in the complement problem is to determine whether
  there is \emph{no} integer $x$ in $[0, x^*)$ such that $f(x) \geq \tau$.
  The idea is to negate all the values to flip the roles of winning and losing.
  
  For each function $g_i$, we construct two sources and wire them to a sink
  in the same way as Theorem~\ref{thm:sim NP-hard}, but negated:
  if $g_i(x_i) = \pm 1$, then we design the packets to have a net effect
  of $\mp 1$ at time $x_i$ and $0$ otherwise.

  Setting $k_0 = \tau-1$, the loading bar goes negative
  (and the player loses)
  exactly when the sum of the functions $g_i$ goes up to~$\tau$, i.e.,
  the Periodic Sum Threshold Problem has a ``yes'' answer.
  Setting $k^* = k_0 - \lambda$, the loading bar reaches $k^*$
  (and the player wins)
  exactly when the sum of the functions $g_i$ goes down to $\lambda$, i.e.,
  the Periodic Sum Threshold Problem has a ``no'' answer.
\end{proof}


\section{Characterizing Perfect Layouts}
\label{sec:perfect-layouts}

Suppose we are given a board which is empty except for the location of the
three data sinks.
Is it possible to place arrows such that all possible input packets
get routed to the correct sink?
We call such a configuration of arrows a \defn{perfect layout}.
In particular, such a layout guarantees victory,
regardless of the data sources.
In this section, we give a full characterization of boards
and sink placements that admit a perfect layout.
Some of our results work for a general number $c$ of colors,
but the full characterization relies on $c=3$.

\subsection{Colors Not Arrows}

We begin by showing that we do not need to consider the directions of the arrows, only their colors and locations in the grid.

Let \(B\) be a board with specified locations of sinks,
and let \(\partial B\) be the set of edges on the boundary of \(B\).
Suppose we are given an assignment of colors to the cells of \(B\)
that agrees with the colors of the sinks;
let \(C_i\) be the set of grid cells colored with color \(i\).
We call two cells of \(C_i\), or a cell of \(C_i\) and a boundary edge
$e \in \partial B$, \defn{visible} to each other if and only if
they are in the same row or the same column
and no sink of a color other than \(i\) is between them.
Let \(G_i\) be the graph whose vertex set is \(C_i \cup \partial B\),
with edges between pairs of vertices that are visible to each other.

\begin{lemma}
  \label{lem:perfect-colors}
  Let \(B\) be a board with specified locations of sinks.
  Then \(B\) admits a perfect layout if and only if it is possible to choose
  colors for the remaining cells of the grid such that, for each color~\(i\),
  the graph \(G_i\) is connected.
\end{lemma}
\begin{proof}
  ($\Longrightarrow$)
  Without loss of generality, assume that
  the perfect layout has the minimum possible number of arrows.
  Color the cells of the board with the same colors as the sinks and arrows in the perfect layout.
  (If a cell is empty in the perfect layout,
  then give it the same color as an adjacent cell;
  this does not affect connectivity.)
  Fix a color \(i\).
  Every boundary edge is connected to the sink of color \(i\)
  by the path a packet of color \(i\) follows when entering from that edge.
  (In particular, the path cannot go through a sink of a different color.)
  By minimality of the number of arrows in the perfect layout,
  every arrow of color \(i\) is included in such a path.
  Therefore \(G_i\) is connected.

  ($\Longleftarrow$)
  We will replace each cell by an arrow of the same color to form a perfect layout.
  Namely, for each color~\(i\), choose a spanning tree of \(G_i\)
  rooted at the sink of color~\(i\),
  and direct arrows from children to parents in this tree.
  By connectivity, any packet entering from a boundary edge
  will be routed to the correct sink, walking up the tree to its root.
\end{proof}

\subsection{Impossible Boards}
\label{sec:impossible-boards}

Next we show that certain boards cannot have perfect layouts.
First we give arguments about boards containing sinks
too close to the boundary or each other.
Then we give an area-based constraint on board size.

\begin{lemma}
  \label{lem:sink-distance}
  If there are fewer than $c-1$ blank cells in a row or column
  between a sink and a boundary of the grid,
  then there is no perfect layout.
\end{lemma}
\begin{proof}
  A perfect layout must prevent packets of the other \(c-1\) colors
  entering at this boundary from reaching this sink;
  this requires enough space for \(c-1\) arrows.
\end{proof}

\begin{lemma}
  \label{lem:impossible-c3}
  For $c=3$, a board has no perfect layout if either
  (as shown in Figure~\ref{fig:impossible})
  \begin{enumerate}[(a)]
    \item a data sink is two cells away from three boundaries and adjacent to another sink;
    \item a data sink is two cells away from two incident boundaries and is adjacent to two other sinks;
    \item a data sink is two cells away from two opposite boundaries and is adjacent to two other sinks; or
    \item a data sink is two cells away from three boundaries and is one blank cell away from a pair of adjacent sinks.
  \end{enumerate}
\end{lemma}

\begin{figure}
  \centering
  \subcaptionbox{}{\includegraphics[scale=0.75]{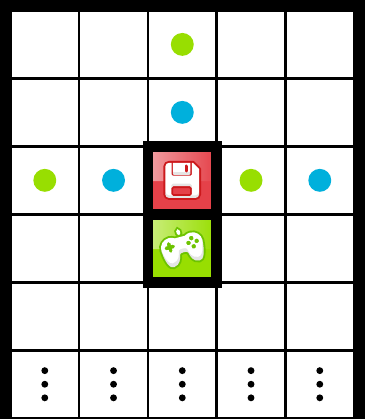}}\hfil
  \subcaptionbox{}{\includegraphics[scale=0.75]{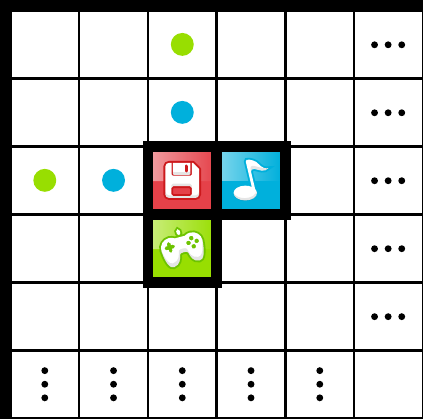}}\hfil
  \subcaptionbox{}{\includegraphics[scale=0.75]{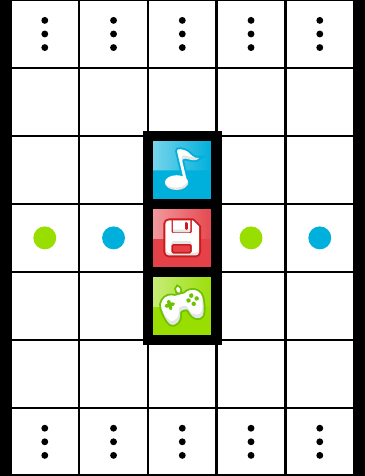}}\hfil
  \subcaptionbox{}{\includegraphics[scale=0.75]{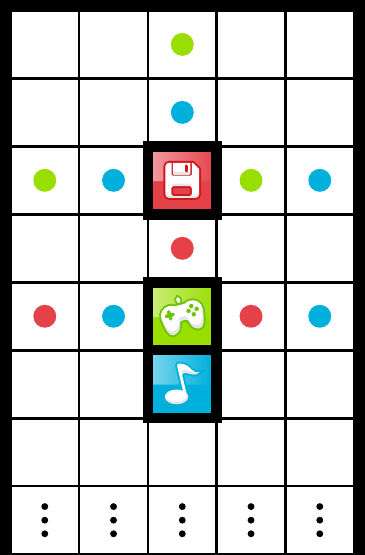}}
  \caption{Sink configurations with no perfect layout.
    Dots indicate arrows of forced colors
    (up to permutation within a row or column).}
  \label{fig:impossible}
\end{figure}

\begin{proof}
  Assume by symmetry that, in each case, the first mentioned sink is red.

  Cases (a), (b), and (c):
  The pairs of cells between the red sink and the boundary
  (marked with dots in the figure) must contain a green arrow and a blue arrow
  to ensure those packets do not reach the red sink.
  Thus there are no available places to place a red arrow
  in the same row or column as the red sink,
  so red packets from other rows or columns cannot reach the red sink.
  
  Case (d): The pairs of cells between the red sink and the boundary
  (marked with green and blue dots in the figure)
  must contain a green arrow and a blue arrow to ensure those packets
  do not collide with the red sink.
  Thus the blank square between the red sink and the other pair of sinks
  must be a red arrow pointing toward the red sink,
  to allow packets from other rows and columns to reach the red sink.
  Assume by symmetry that the sink nearest the red sink is green.
  As in the other cases, the pairs of cells between the green sink and the
  boundary must be filled with red and blue arrows.
  Thus there are no green arrows to route green packets from other
  rows or columns to the green sink.
\end{proof}

We now prove a constraint on the sizes of boards that admit a perfect layout.

\begin{lemma}
  \label{lem:size-constraint}
  Let \(c\) be the number of colors.
  Suppose there is a perfect layout on a board where \(m\) and \(n\) are respectively the number of rows and columns, and \(p\) and \(q\) are respectively the number of rows and columns that contain at least one sink.
  Then
  \begin{equation}
    \label{eqn:size-constraint}
    c(m + n) + (c-2)(p + q) \le m n - c.
  \end{equation}
\end{lemma}
\begin{proof}
  Each of the \(m - p\) unoccupied rows must contain \(c\) vertical arrows in order to redirect packets of each color out of the row.
  Each of the \(p\) occupied rows must contain \(c-1\) vertical arrows to the
  left of the leftmost sink in order to redirect incorrectly colored packets
  from the left boundary edge away from that sink;
  similarly, there must be \(c-1\) vertical arrows to the right of the
  rightmost sink.
  Thus we require \(c(m - p) + 2(c - 1)p = c m + (c - 2)p\)
  vertical arrows overall.
  By the same argument, we must have \(c n + (c - 2)q\) horizontal arrows,
  for a total of
  \(c(m + n) + (c - 2)(p + q)\)
  arrows.
  There are \(m n - c\) cells available for arrows, which proves the claim.
\end{proof}

Up to insertion of empty rows or columns, rotations, reflections, and
recolorings, there are six different configurations that $c=3$ sinks
may have with respect to each other, shown in Figure~\ref{fig:3-sink-configs}.
We define a board's \defn{type} according to this configuration of its sinks
(C, I, J, L, Y, or /).

\begin{figure}[htbp]
  \centering
  \hfil
  \subcaptionbox{C}{\includegraphics[scale=0.9]{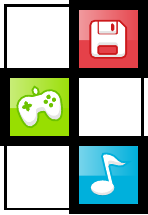}}\hfil
  \subcaptionbox{I}{\includegraphics[scale=0.9]{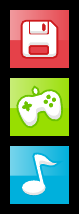}}\hfil
  \subcaptionbox{J}{\includegraphics[scale=0.9]{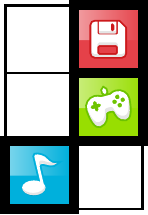}}\hfil
  \subcaptionbox{L}{\includegraphics[scale=0.9]{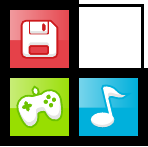}}\hfil
  \subcaptionbox{Y}{\includegraphics[scale=0.9]{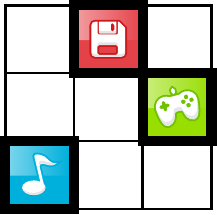}}\hfil
  \subcaptionbox{/}{\includegraphics[scale=0.9]{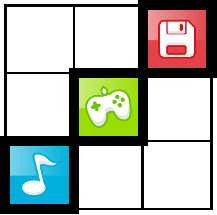}}%
  \hfil
  \caption{The six possible configurations of three sinks up to rotations, reflections, recolorings, and removal of empty rows.}
  \label{fig:3-sink-configs}
\end{figure}

A board's type determines the values of \(p\) and \(q\)
and thus the minimal board sizes as follows.
Define a board to have size \defn{at least} $m \times n$ if
it has at least $m$ rows and at least $n$ columns, or vice versa.

\begin{lemma}
  \label{lem:size-constraint-c3}
  For a perfect layout to exist with \(c=3\), it is necessary that:
  \begin{itemize}
  \item Boards of type Y or / have size at least $7\times8$.
  \item Boards of type C or J have size at least $7\times8$ or $6\times9$.
  \item Boards of type L have size at least $7\times7$ or $6\times9$.
  \item Boards of type I have size at least $7\times7$, $6\times9$, or $5\times11$.
  \end{itemize}
\end{lemma}
\begin{proof}
  These bounds follow from Lemma~\ref{lem:size-constraint} together with the requirement from Lemma~\ref{lem:sink-distance} that it be possible to place sinks at least two cells away from the boundary.
\end{proof}

\subsection{Constructing Perfect Layouts}

In this section, we complete our characterization of boards
with perfect layouts for \(c=3\).
We show that Lemmas \ref{lem:sink-distance}, \ref{lem:impossible-c3}, and
\ref{lem:size-constraint-c3} are the only obstacles to a perfect layout:

\begin{theorem}
  \label{thm:characterization-c3}
  A board with $c=3$ sinks has a perfect layout if and only if the following conditions all hold:
  \begin{enumerate}
  \item All sinks are at least two cells away from the boundary
  (Lemma~\ref{lem:sink-distance}).
  \item The board does not contain any of the four unsolvable configurations in Figure~\ref{fig:impossible} (Lemma~\ref{lem:impossible-c3}).
  \item The board obeys the size bounds of Lemma~\ref{lem:size-constraint-c3}.
  \end{enumerate}
\end{theorem}

We call a board \defn{minimal} if it has one of the minimal dimensions
for its type as defined in Lemma~\ref{lem:size-constraint-c3}.
Our strategy for proving Theorem~\ref{thm:characterization-c3}
will be to reduce the problem to the finite set of minimal boards,
which we then verify by computer.
We will accomplish this by removing empty rows and columns
from non-minimal boards to reduce their size,
which we show can always be done while preserving the above conditions.

\begin{lemma}
  \label{lem:characterization-c3-minimal}
  All minimal boards satisfying the three conditions of
  Theorem~\ref{thm:characterization-c3} have a perfect layout.
\end{lemma}
\begin{proof}
  \renewcommand{\qedsymbol}{\(\blacksquare\)}
  The proof is by exhaustive computer search of all such minimal boards.
  We wrote a Python program to generate all possible board patterns,
  reduce each perfect layout problem to Satisfiability Modulo Theories
  (SMT), and then solve it using Z3 \cite{z3}.
  The results of this search are in Appendix~\ref{apx:computer-solutions}.
\end{proof}

If \(B_0\) and \(B_1\) are boards, then we define \defn{\(B_0 \pmb{\lessdot} B_1\)}
to mean that \(B_0\) can be obtained by removing a single empty row or column
from \(B_1\).

\begin{lemma}
  \label{lem:add-row}
  If \(B_0 \lessdot B_1\) and \(B_0\) has a perfect layout,
  then \(B_1\) also has a perfect layout.
\end{lemma}
\begin{proof}
  By symmetry, consider the case where $B_1$ has an added row.
  By Lemma~\ref{lem:perfect-colors}, it suffices to show that we can color
  the cells of the new row while preserving connectivity in each color.
  We do so by duplicating the colors of the cells (including sinks)
  in an adjacent row.
  Connectivity of the resulting coloring follows from that of the original.
\end{proof}

\begin{lemma}
  \label{lem:characterization-c3-nonminimal}
  Let \(B_1\) be a non-minimal board satisfying the three conditions
  of Theorem~\ref{thm:characterization-c3}.
  Then there exists a board \(B_0\) that also satisfies all three conditions
  and such that \(B_0 \lessdot B_1\).
\end{lemma}
\begin{proof}
  By symmetry, suppose $B_1$ is non-minimal in its number $m$ of rows.
  By removing a row from $B_1$ that is not among the first or last two rows
  and does not contain a sink, we obtain a board \(B'_0\) satisfying conditions
  (1) and (3) such that \(B'_0 \lessdot B_1\).
  If \(B'_0\) also satisfies condition (2), then we are done,
  so we may assume that it does not.

  Then \(B'_0\) must contain one of the four unsolvable configurations,
  and \(B_1\) is obtained by inserting a single empty row or column to
  remove the unsolvable configuration.
  Figure~\ref{fig:perfect-reductions} shows all possibilities for \(B'_0\),
  as well as the locations where rows or columns may be inserted to yield
  a corresponding possibility for \(B_1\).
  (\(B'_0\) may have additional empty rows and columns beyond those shown,
  but this does not affect the proof.)
  For each such possibility, Figure~\ref{fig:perfect-reductions} highlights
  another row or column which may be deleted from \(B_1\) to yield
  \(B_0 \lessdot B_1\) where \(B_0\) satisfies all three conditions.
\end{proof}

\begin{figure}
  \centering
  \hfil
  \subcaptionbox{$L$, $7\times7$}{\includegraphics[scale=0.6]{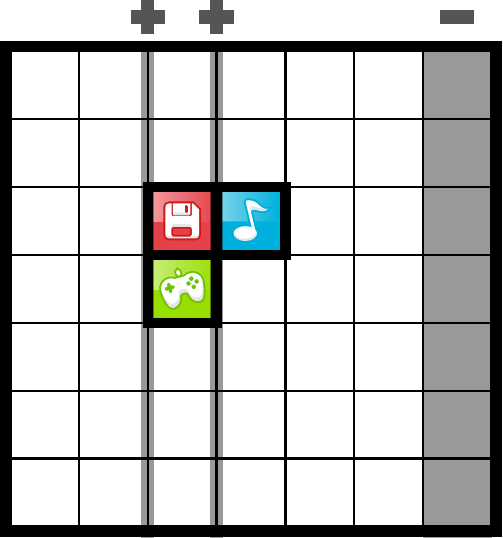}}\hfill
  \subcaptionbox{$L$, $6\times9$}{\includegraphics[scale=0.6]{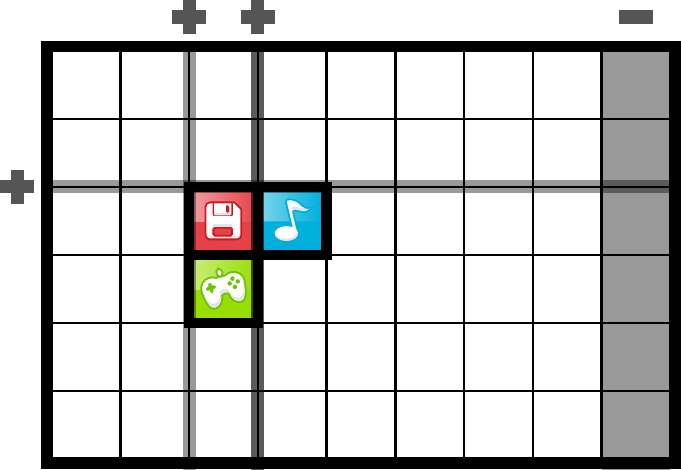}}\hfill
  \subcaptionbox{\label{fig:perfect-reductions-row-choice}$I$, $5\times11$}{\includegraphics[scale=0.6]{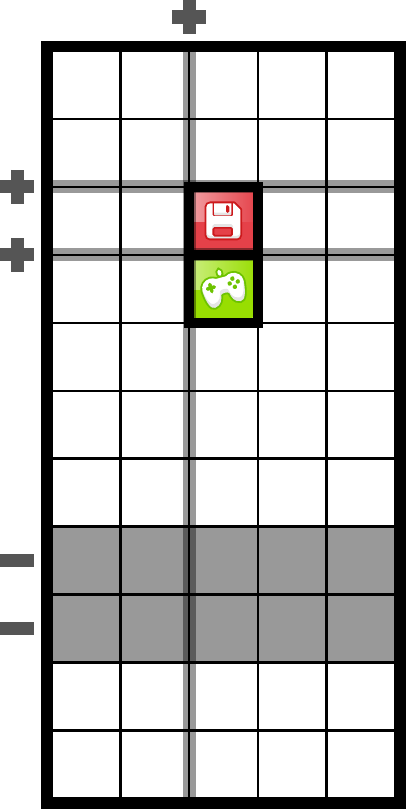}}\hfill
  \subcaptionbox{$I$, $5\times11$}{\includegraphics[scale=0.6]{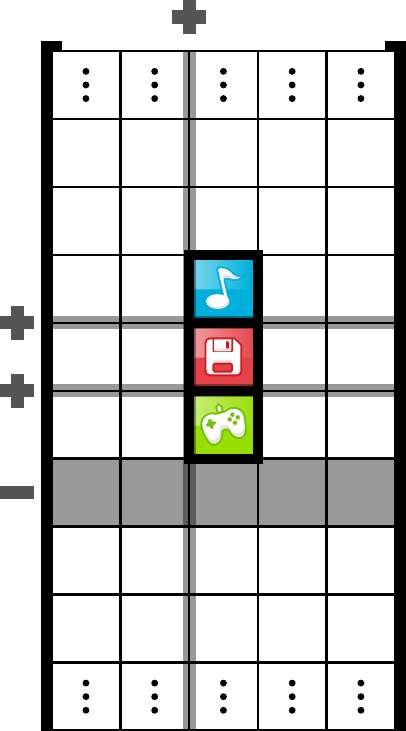}\vspace{0.2in}}\hfill
  \subcaptionbox{$I$, $5\times11$}{\includegraphics[scale=0.6]{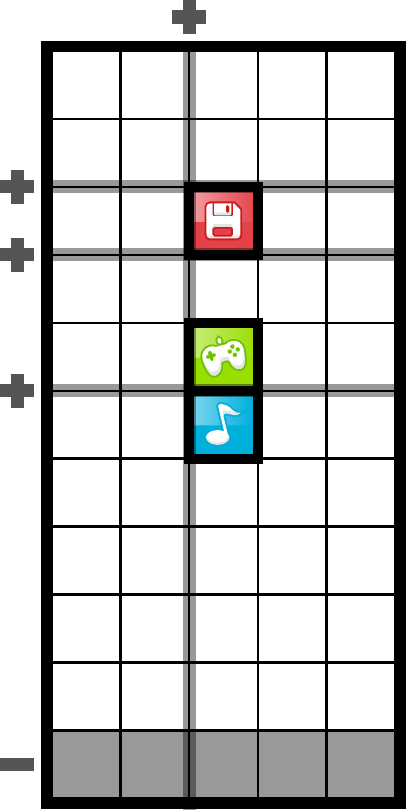}}%
  \hfil
  \caption{
    All boards satisfying conditions (1) and (3) but not (2),
    up to rotations, reflections, and recolorings.
    An empty row or column may be inserted in any of the locations
    marked ``$+$'' to yield a board satisfying all three conditions.
    Removing the row or column marked ``$-$'' then preserves the conditions.
    In case~(c),
    remove a row that does not contain the blue sink.
    In case~(d),
    $\vcenter{\hbox{\includegraphics[scale=0.6]{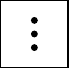}}}$
    denotes zero or more rows.
  }
  \label{fig:perfect-reductions}
\end{figure}

\begin{proof}[Proof of Theorem~\ref{thm:characterization-c3}]
  It follows from Lemmas \ref{lem:sink-distance}, \ref{lem:impossible-c3},
  and \ref{lem:size-constraint-c3} that all boards with perfect layouts
  must obey the three properties of the theorem.
  We prove that the properties are also sufficient
  by induction on the size of the board.
  As a base case, the claim holds for minimal boards
  by Lemma~\ref{lem:characterization-c3-minimal}.
  For non-minimal boards \(B_1\),
  Lemma~\ref{lem:characterization-c3-nonminimal} shows that
  there is a smaller board \(B_0\)
  that satisfies all three conditions and such that \(B_0 \lessdot B_1\).
  By the inductive hypothesis, \(B_0\) has a perfect layout.
  Lemma~\ref{lem:add-row} shows that \(B_1\) also has a perfect layout.
\end{proof}


\section{Open Questions}

The main complexity open question is whether {\gamename} is $\Sigma_2^P$-complete.
Given our NP- and coNP-hardness results, we suspect that this is true.

One could also ask complexity questions of more restrictive versions of the game.
For example, what if the board has a constant number of rows?

When characterizing perfect layouts, we saw many of our lemmas generalized to different numbers of colors. It may be interesting to further explore the game and try to characterize perfect layouts with more than three colors.

A related problem is which boards and configurations of sinks
admit a \defn{damage-free} layout, where any packet entering from
the boundary either reaches the sink of the correct color
or ends up in an infinite loop.  Such a layout avoids losing,
and in the game as implemented, such a layout actually wins the game
(because the player wins if there is ever insufficient room
for a new source to be placed).
Can we characterize such layouts like we did for perfect layouts?

Perfect and damage-free layouts are robust to any possible sources.
However, for those boards that do not admit a perfect or damage-free layout,
it would be nice to have an algorithm that determines whether a given
set of sources or sequence of packets still has a placement of arrows
that will win on that board.
Because the board starts empty except for the sinks,
our hardness results do not apply.

Having a unique solution is often a desirable property of puzzles. Thus it is natural to ask about ASP-hardness and whether counting the number of solutions is \#P-hard.


\section*{Acknowledgments}

This work was initiated during open problem solving in the MIT class on
Algorithmic Lower Bounds: Fun with Hardness Proofs (6.892)
taught by Erik Demaine in Spring 2019.
We thank the other participants of that class
for related discussions and providing an inspiring atmosphere.
In particular, we thank Quanquan C. Liu for helpful discussions
and contributions to early results.

Most figures of this paper were drawn using SVG Tiler
[\url{https://github.com/edemaine/svgtiler}].
Icons (which match the game) are by looneybits and
released in the public domain
[\url{https://opengameart.org/content/gui-buttons-vol1}].

\bibliographystyle{alpha}
\bibliography{biblio}


\appendix

\section{Perfect Layouts from the Automated Solver}
\label{apx:computer-solutions}

The following 13 figures show all cases found by the automated solver.
Figures~\ref{fig:solver-solutions-C-6x9},
\ref{fig:solver-solutions-C-7x8}, and
\ref{fig:solver-solutions-C-8x7} correspond to sinks in the C pattern.
Figures~\ref{fig:solver-solutions-I-5x11},
\ref{fig:solver-solutions-I-6x9}, and
\ref{fig:solver-solutions-I-7x7} correspond to sinks in the I pattern.
Figures~\ref{fig:solver-solutions-J-6x9},
\ref{fig:solver-solutions-J-7x8}, and
\ref{fig:solver-solutions-J-8x7} correspond to sinks in the J pattern.
Figures~\ref{fig:solver-solutions-L-6x9} and
\ref{fig:solver-solutions-L-7x7} correspond to sinks in the L pattern.
Figure~\ref{fig:solver-solutions-Y-7x8} corresponds to sinks in the Y pattern.
Finally,
Figure~\ref{fig:solver-solutions-/-7x8} corresponds to sinks in the / pattern.

\def\scale{0.8}
\tabcolsep=1em

\begin{figure}[bp]
  \centering
  \begin{tabular}{ccc}
    \subcaptionbox{C, $6\times9$}{\includegraphics[scale=\scale]{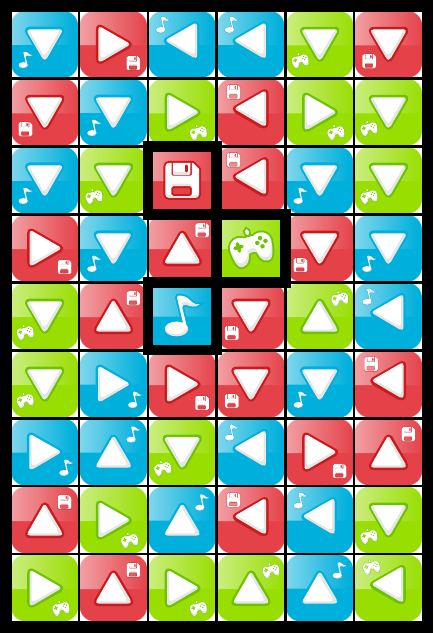}} &
    \subcaptionbox{C, $6\times9$}{\includegraphics[scale=\scale]{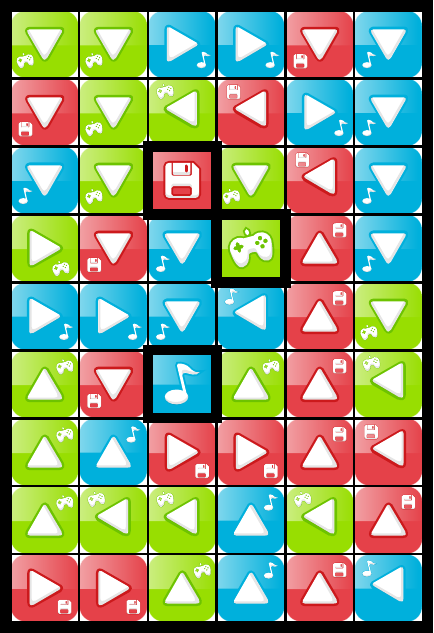}} &
    \subcaptionbox{C, $6\times9$}{\includegraphics[scale=\scale]{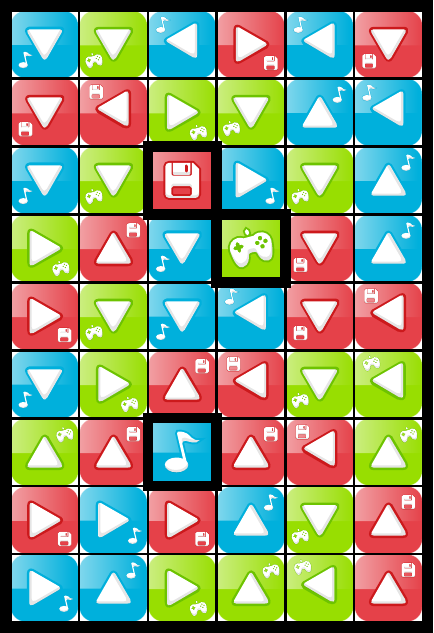}}
    \smallskip\\
    \subcaptionbox{C, $6\times9$}{\includegraphics[scale=\scale]{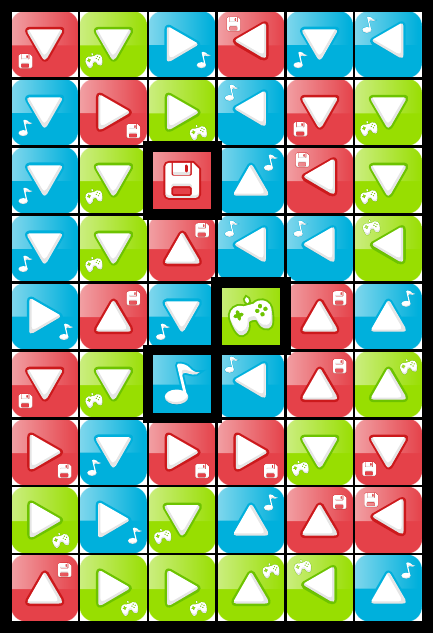}} &
    \subcaptionbox{C, $6\times9$}{\includegraphics[scale=\scale]{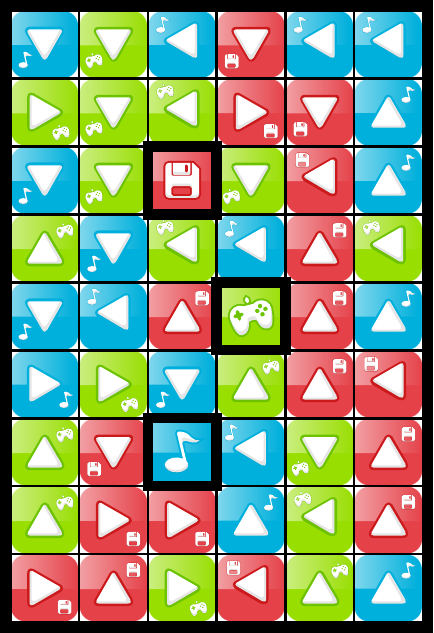}} &
    \subcaptionbox{C, $6\times9$}{\includegraphics[scale=\scale]{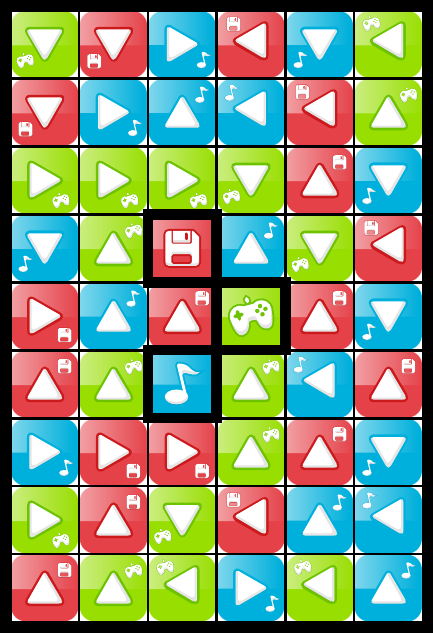}}
  \end{tabular}
  \caption{Solutions from the automated solver for sinks in the C pattern of size $6\times9$.}
  \label{fig:solver-solutions-C-6x9}
\end{figure}

\begin{figure}
  \centering
  \begin{tabular}{ccc}
    \smallskip\\
    \subcaptionbox{C, $7\times8$}{\includegraphics[scale=\scale]{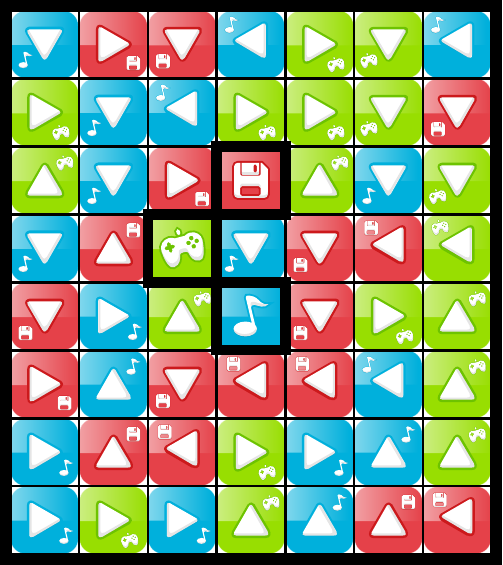}} &
    \subcaptionbox{C, $7\times8$}{\includegraphics[scale=\scale]{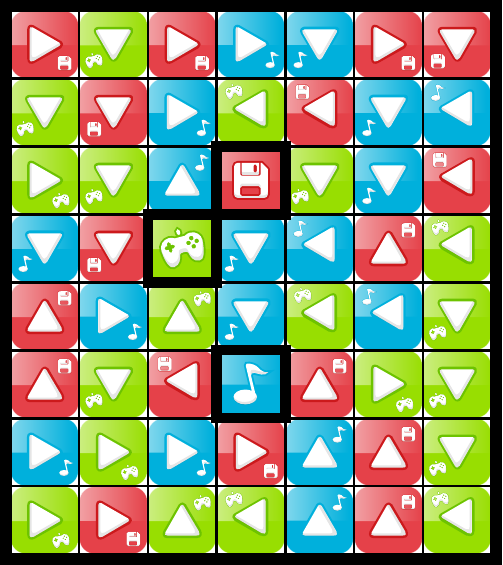}} &
    \subcaptionbox{C, $7\times8$}{\includegraphics[scale=\scale]{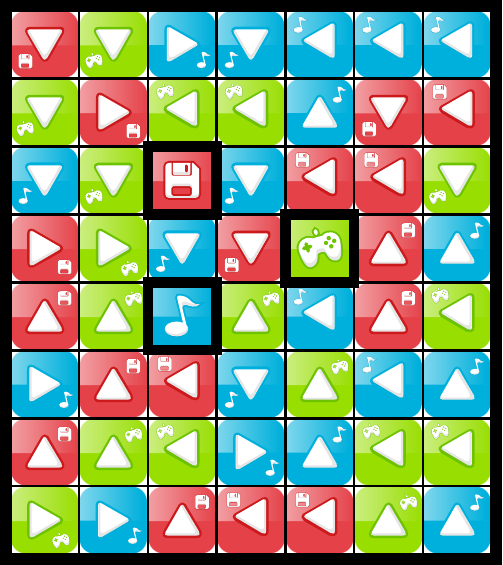}}
    \smallskip\\
    \subcaptionbox{C, $7\times8$}{\includegraphics[scale=\scale]{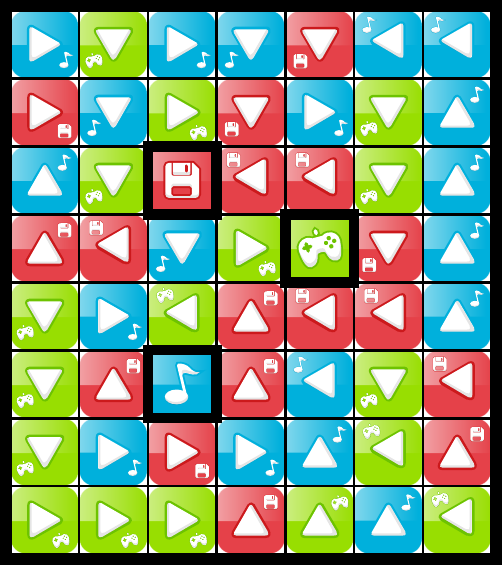}} &
    \subcaptionbox{C, $7\times8$}{\includegraphics[scale=\scale]{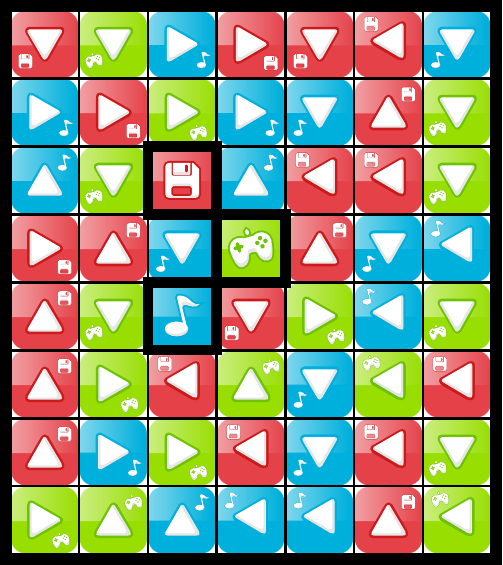}} &
    \subcaptionbox{C, $7\times8$}{\includegraphics[scale=\scale]{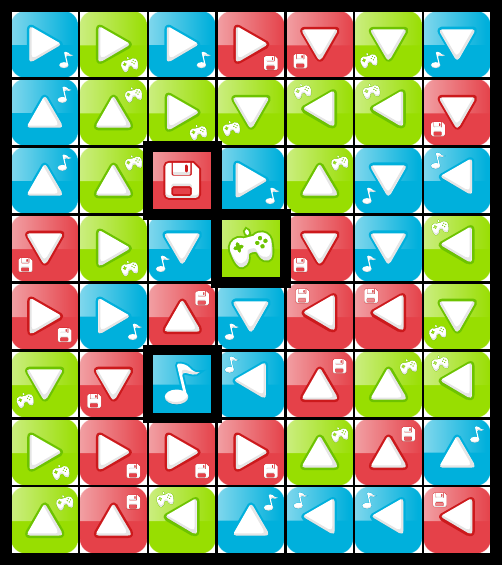}}
  \end{tabular}
  \caption{Solutions from the automated solver for sinks in the C pattern of size $7\times8$.}
  \label{fig:solver-solutions-C-7x8}
\end{figure}

\begin{figure}
  \centering
  \begin{tabular}{ccc}
    \subcaptionbox{C, $8\times7$}{\includegraphics[scale=\scale]{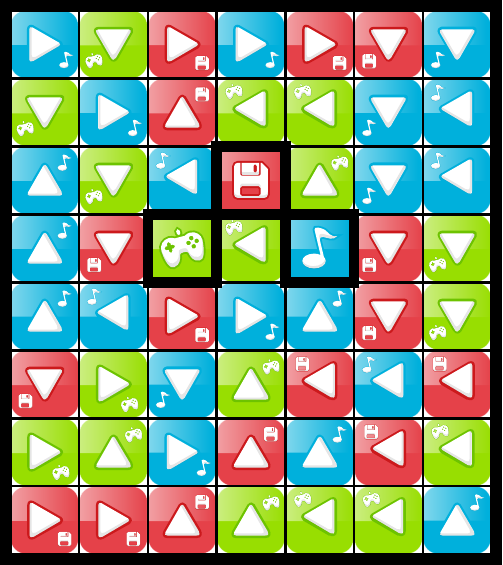}} &
    \subcaptionbox{C, $8\times7$}{\includegraphics[scale=\scale]{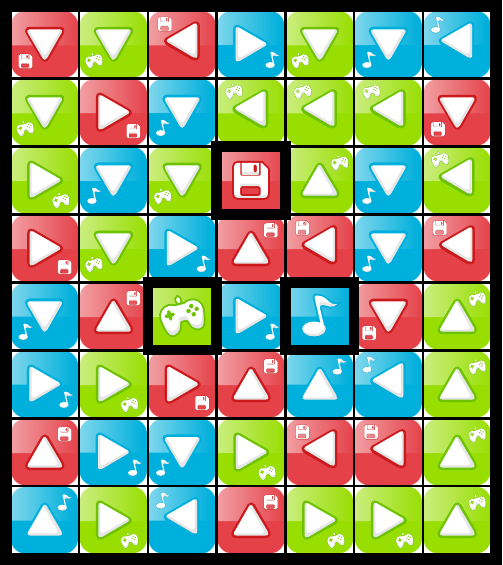}} &
    \subcaptionbox{C, $8\times7$}{\includegraphics[scale=\scale]{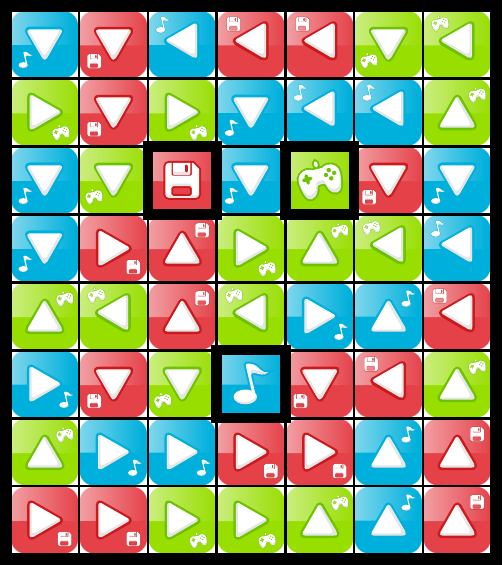}}
    \smallskip\\
    \subcaptionbox{C, $8\times7$}{\includegraphics[scale=\scale]{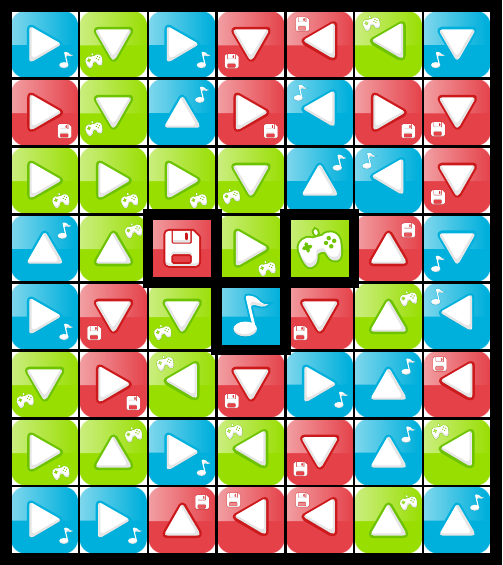}} &
    \subcaptionbox{C, $8\times7$}{\includegraphics[scale=\scale]{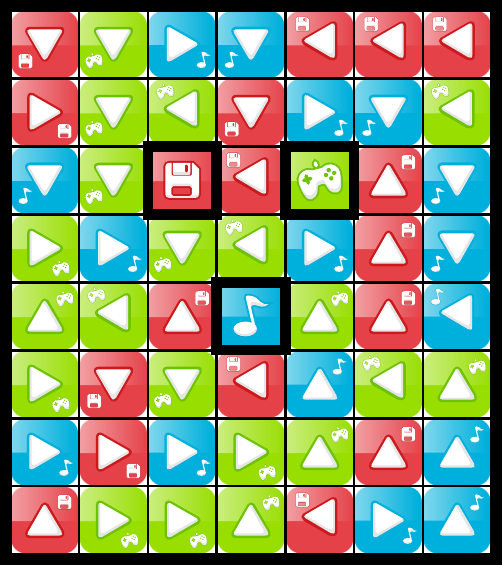}} &
    \subcaptionbox{C, $8\times7$}{\includegraphics[scale=\scale]{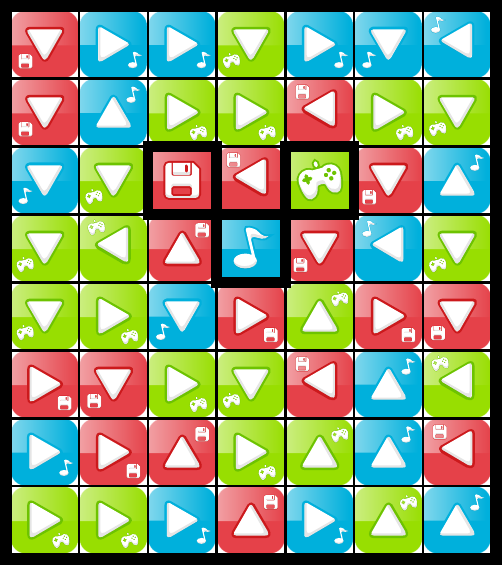}}
  \end{tabular}
  \caption{Solutions from the automated solver for sinks in the C pattern of size $8\times7$.}
  \label{fig:solver-solutions-C-8x7}
\end{figure}

\begin{figure}
  \centering
  \begin{tabular}{ccccc}
    \subcaptionbox{I, $5\times11$}{\includegraphics[scale=\scale]{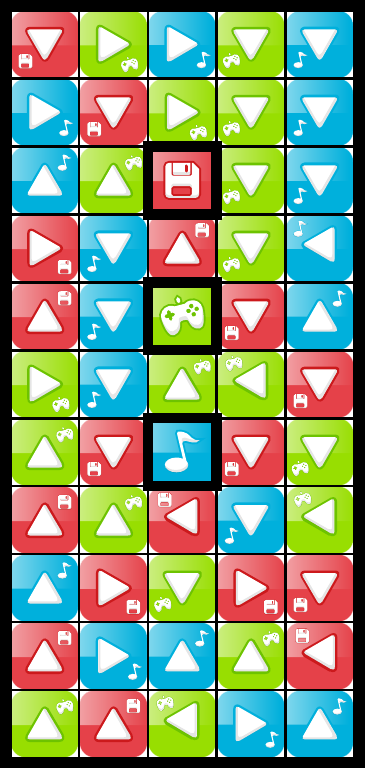}} &
    \subcaptionbox{I, $5\times11$}{\includegraphics[scale=\scale]{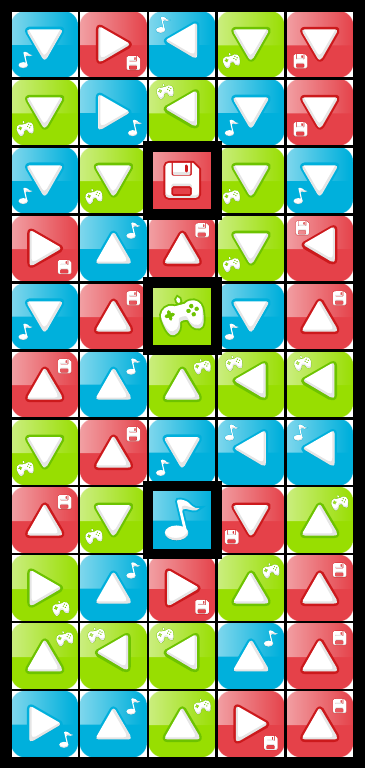}} &
    \subcaptionbox{I, $5\times11$}{\includegraphics[scale=\scale]{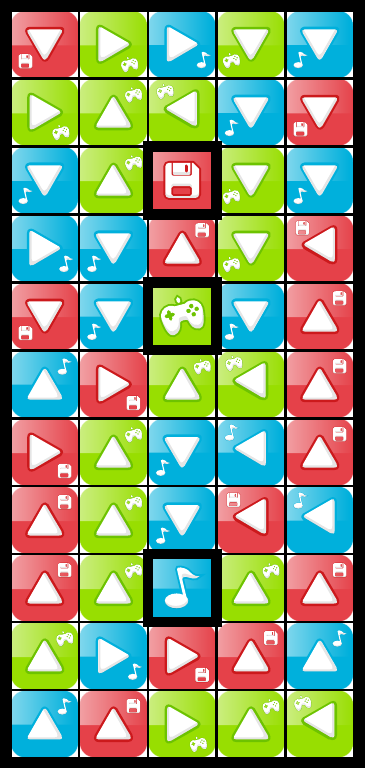}} &
    \subcaptionbox{I, $5\times11$}{\includegraphics[scale=\scale]{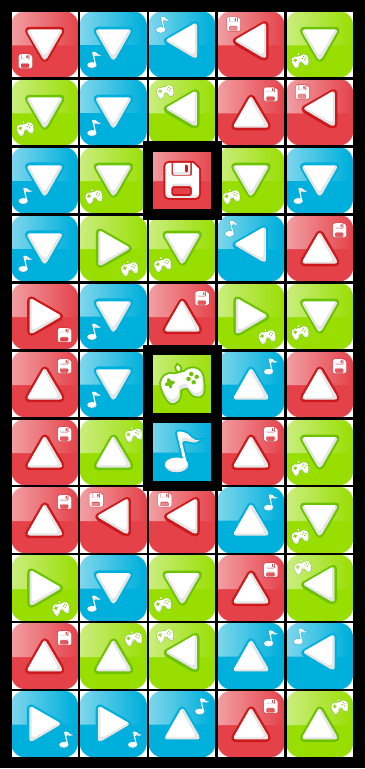}}
    \smallskip\\
    \subcaptionbox{I, $5\times11$}{\includegraphics[scale=\scale]{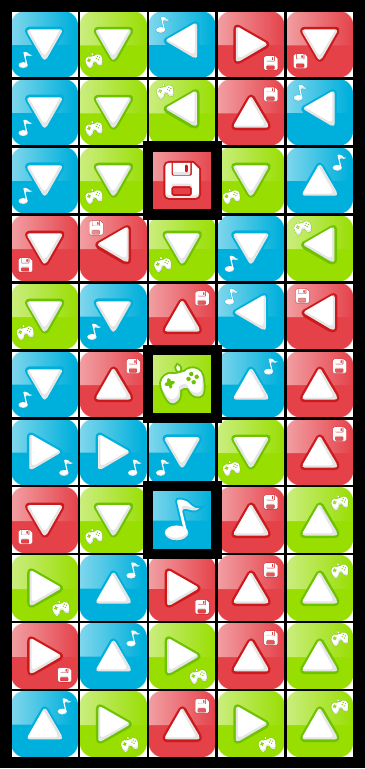}} &
    \subcaptionbox{I, $5\times11$}{\includegraphics[scale=\scale]{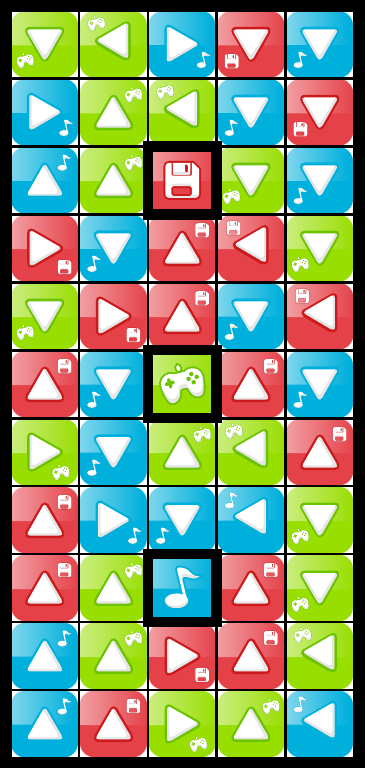}} &
    \subcaptionbox{I, $5\times11$}{\includegraphics[scale=\scale]{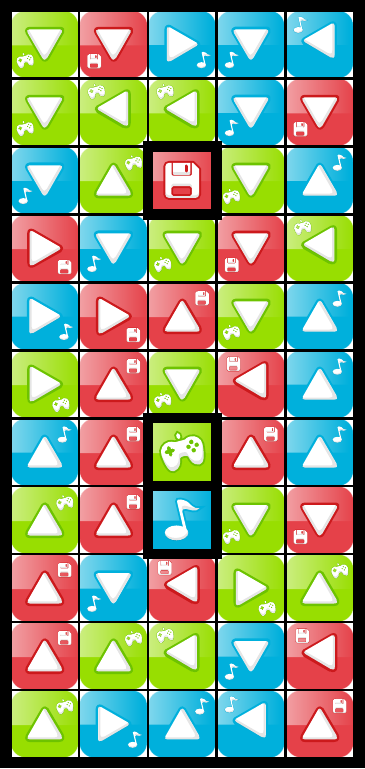}} &
    \subcaptionbox{I, $5\times11$}{\includegraphics[scale=\scale]{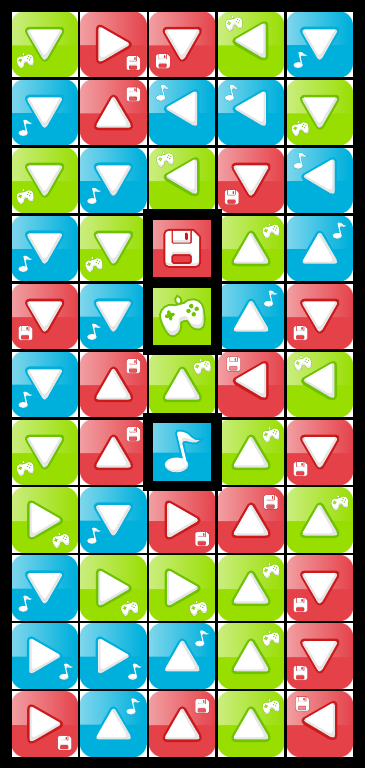}}
    \smallskip\\
    \subcaptionbox{I, $5\times11$}{\includegraphics[scale=\scale]{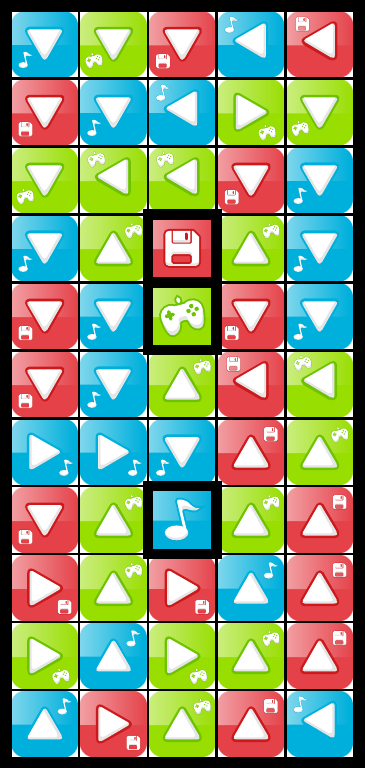}} &
    \subcaptionbox{I, $5\times11$}{\includegraphics[scale=\scale]{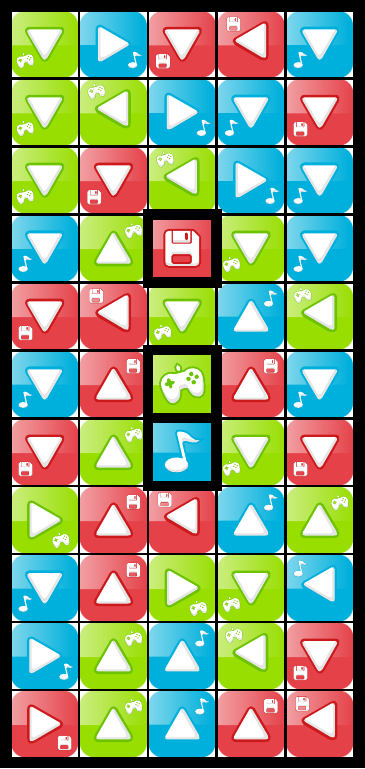}} &
    \subcaptionbox{I, $5\times11$}{\includegraphics[scale=\scale]{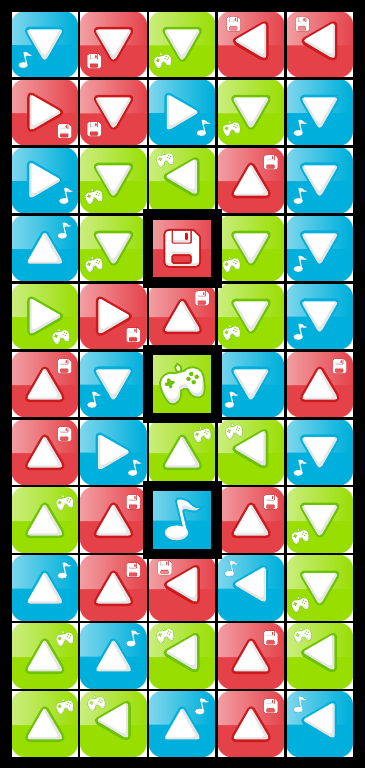}} &
  \end{tabular}
  \caption{Solutions from the automated solver for sinks in the I pattern
    of size $5\times11$.}
  \label{fig:solver-solutions-I-5x11}
\end{figure}

\begin{figure}
  \centering
  \begin{tabular}{ccc}
    \subcaptionbox{I, $6\times9$}{\includegraphics[scale=\scale]{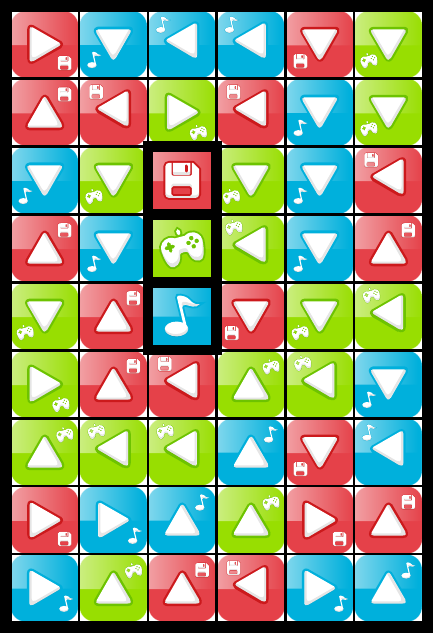}} &
    \subcaptionbox{I, $6\times9$}{\includegraphics[scale=\scale]{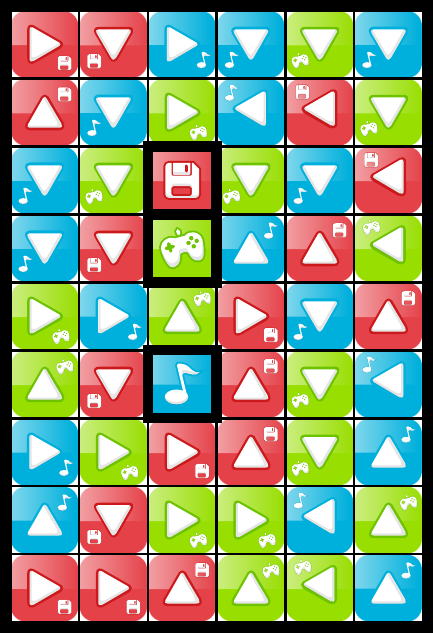}} &
    \subcaptionbox{I, $6\times9$}{\includegraphics[scale=\scale]{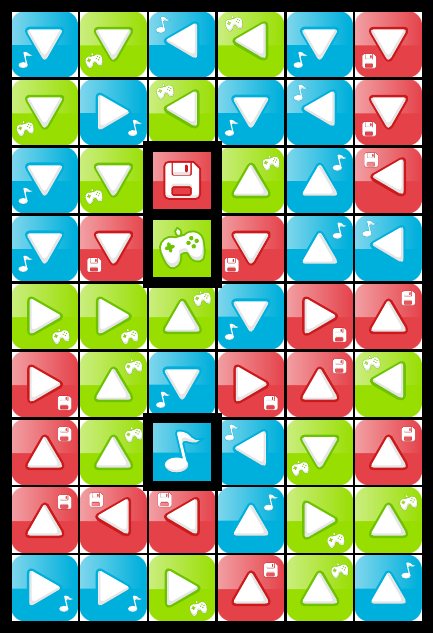}}
    \smallskip\\
    \subcaptionbox{I, $6\times9$}{\includegraphics[scale=\scale]{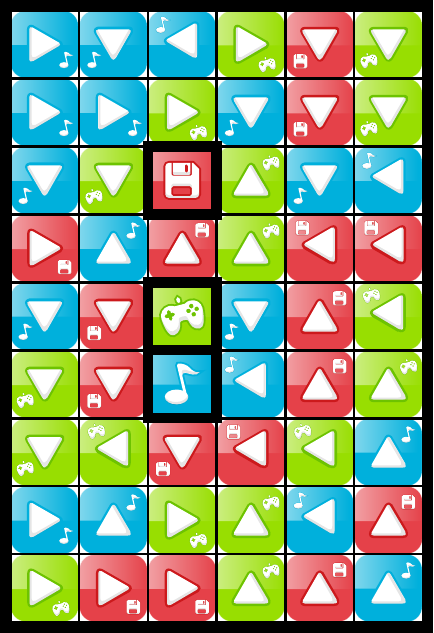}} &
    \subcaptionbox{I, $6\times9$}{\includegraphics[scale=\scale]{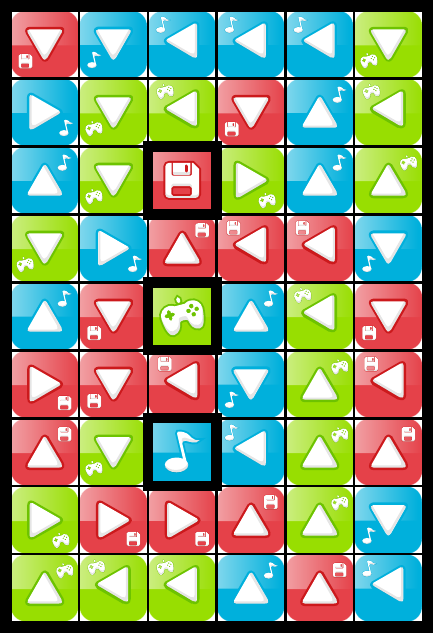}} &
    \subcaptionbox{I, $6\times9$}{\includegraphics[scale=\scale]{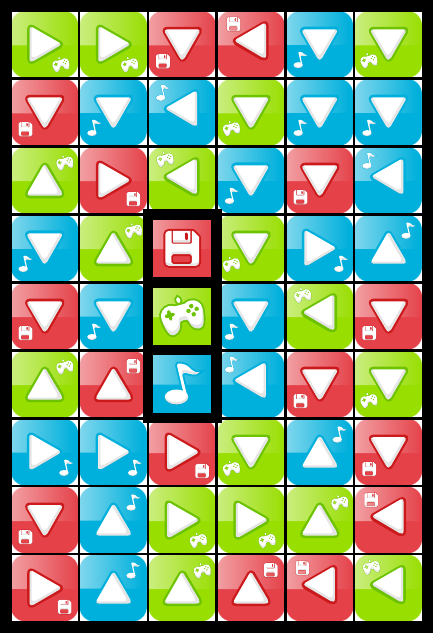}}
  \end{tabular}
  \caption{Solutions from the automated solver for sinks in the I pattern
    of size $6\times9$.}
  \label{fig:solver-solutions-I-6x9}
\end{figure}

\begin{figure}
  \centering
  \begin{tabular}{cc}
    \subcaptionbox{I, $7\times7$}{\includegraphics[scale=\scale]{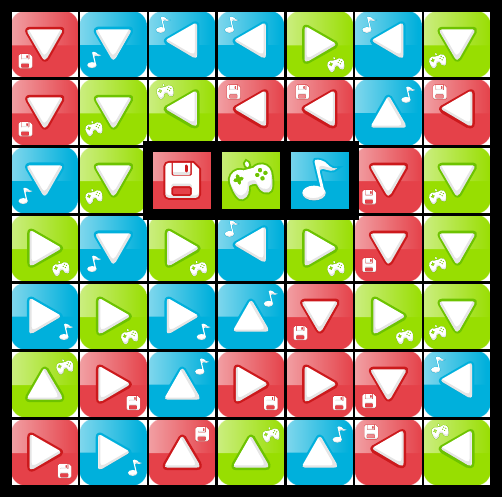}} &
    \subcaptionbox{I, $7\times7$}{\includegraphics[scale=\scale]{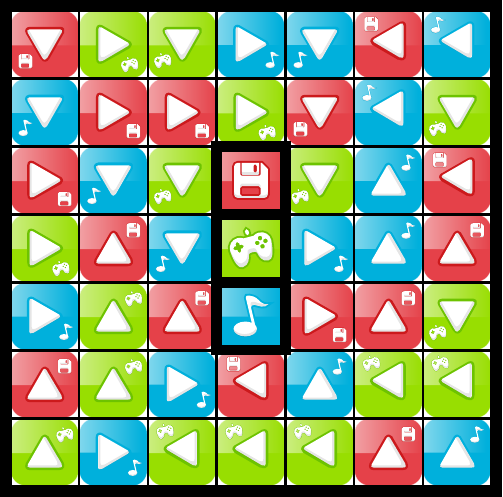}}
  \end{tabular}
  \caption{Solutions from the automated solver for sinks in the I pattern
    of size $7\times7$.}
  \label{fig:solver-solutions-I-7x7}
\end{figure}

\begin{figure}
  \centering
  \tabcolsep=0.75em
  \begin{tabular}{cccc}
    \subcaptionbox{J, $6\times9$}{\includegraphics[scale=\scale]{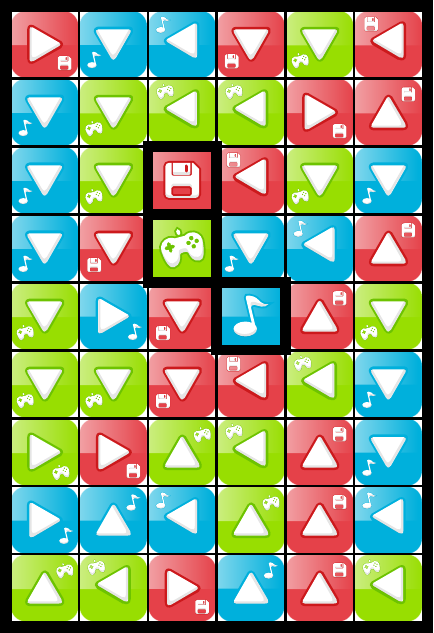}} &
    \subcaptionbox{J, $6\times9$}{\includegraphics[scale=\scale]{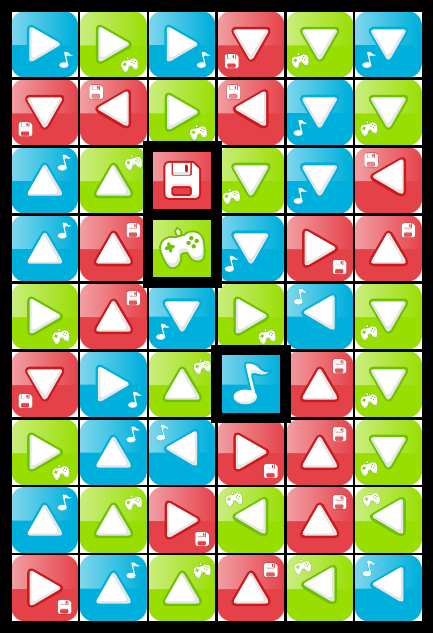}} &
    \subcaptionbox{J, $6\times9$}{\includegraphics[scale=\scale]{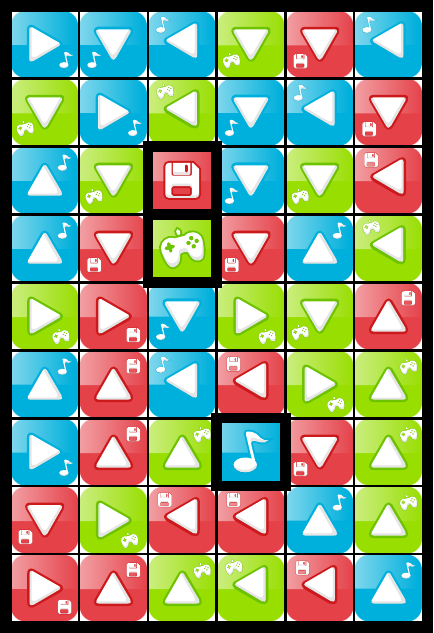}} &
    \subcaptionbox{J, $6\times9$}{\includegraphics[scale=\scale]{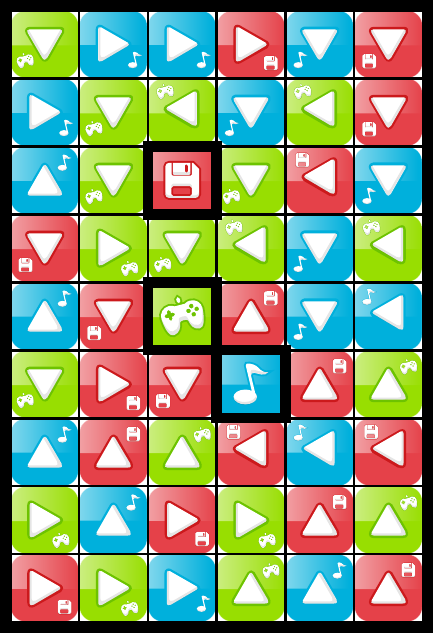}}
    \smallskip\\
    \subcaptionbox{J, $6\times9$}{\includegraphics[scale=\scale]{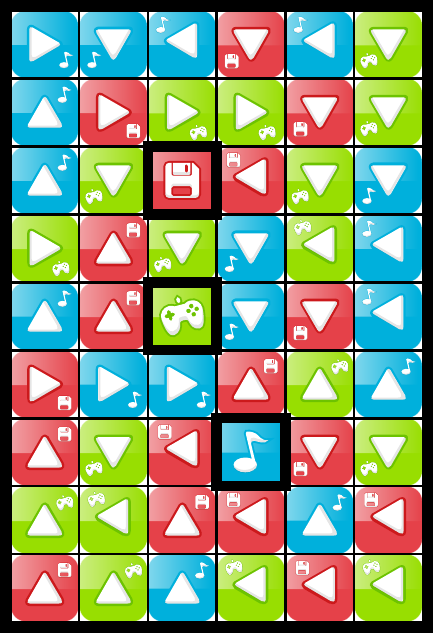}} &
    \subcaptionbox{J, $6\times9$}{\includegraphics[scale=\scale]{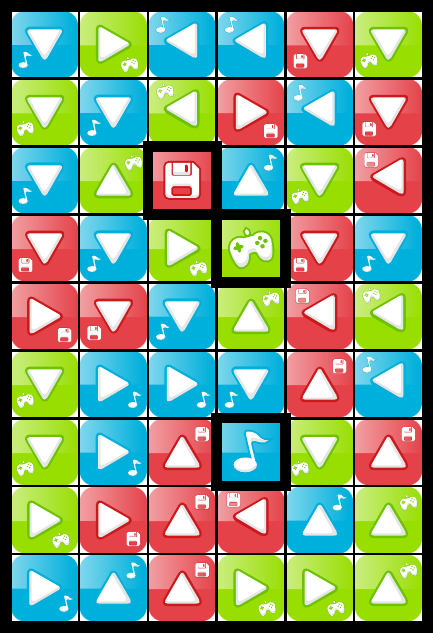}} &
    \subcaptionbox{J, $6\times9$}{\includegraphics[scale=\scale]{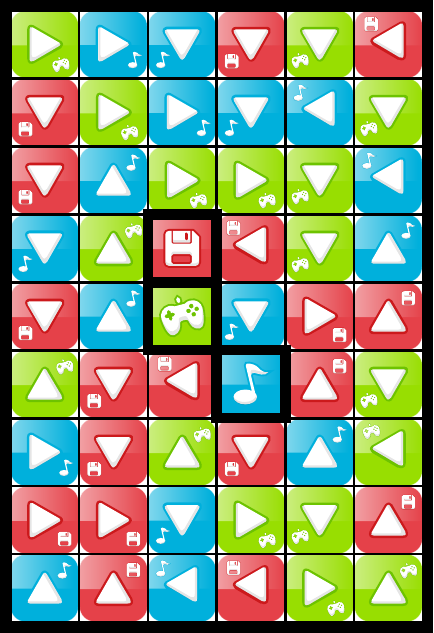}} &
    \subcaptionbox{J, $6\times9$}{\includegraphics[scale=\scale]{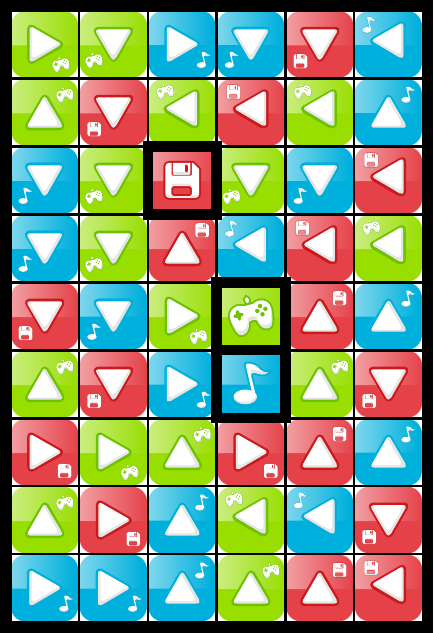}}
    \smallskip\\
    &
    \subcaptionbox{J, $6\times9$}{\includegraphics[scale=\scale]{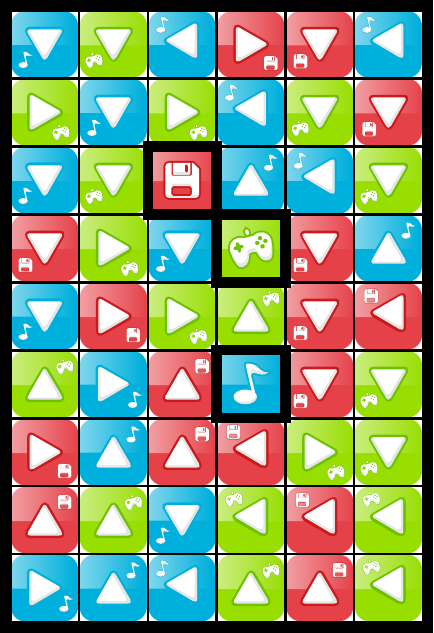}} &
    \subcaptionbox{J, $6\times9$}{\includegraphics[scale=\scale]{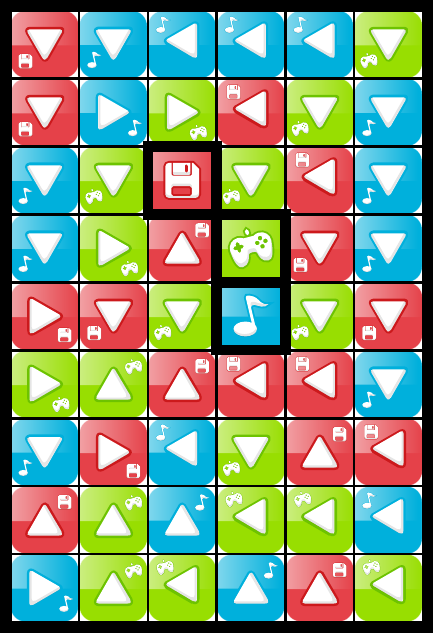}} &
  \end{tabular}
  \caption{Solutions from the automated solver for sinks in the J pattern
    of size $6\times9$.}
  \label{fig:solver-solutions-J-6x9}
\end{figure}

\begin{figure}
  \centering
  \tabcolsep=0.75em
  \begin{tabular}{ccc}
    \subcaptionbox{J, $7\times8$}{\includegraphics[scale=\scale]{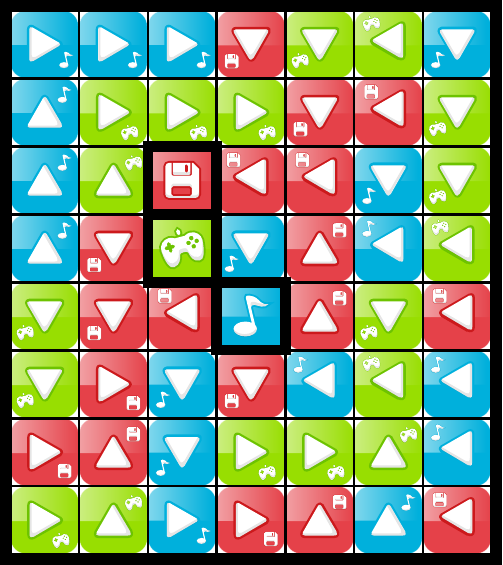}} &
    \subcaptionbox{J, $7\times8$}{\includegraphics[scale=\scale]{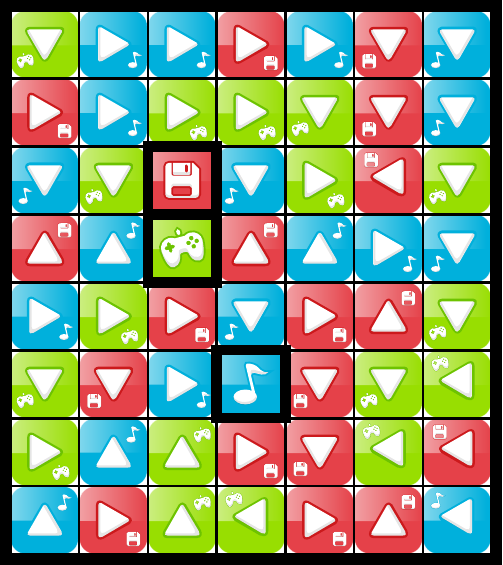}} &
    \subcaptionbox{J, $7\times8$}{\includegraphics[scale=\scale]{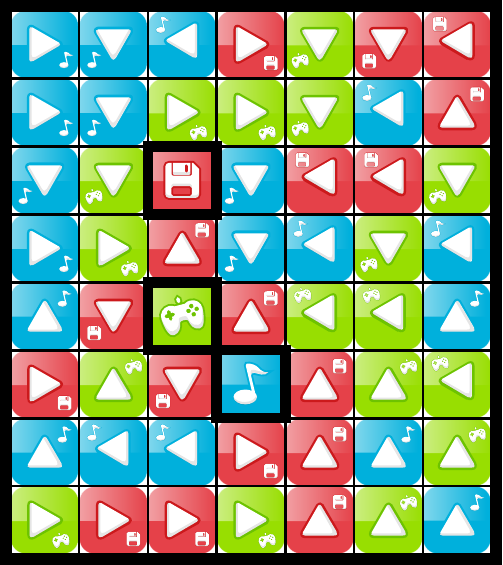}}
    \smallskip\\
    \subcaptionbox{J, $7\times8$}{\includegraphics[scale=\scale]{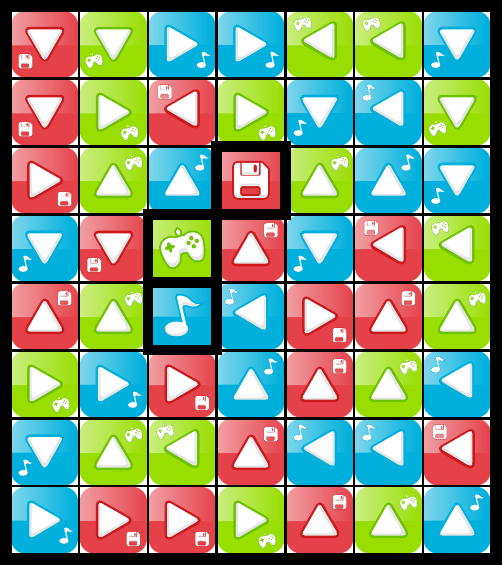}} &
    \subcaptionbox{J, $7\times8$}{\includegraphics[scale=\scale]{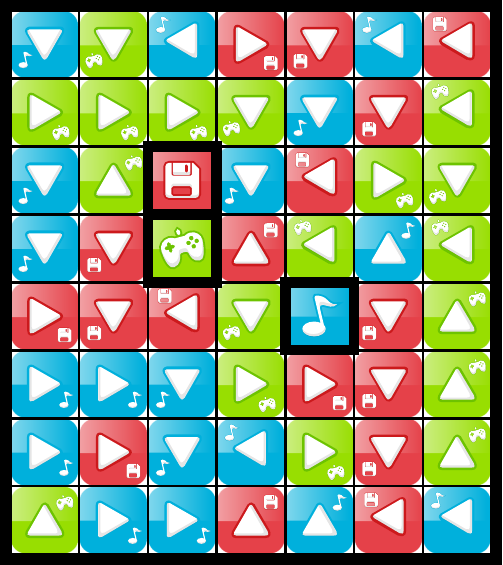}} &
    \subcaptionbox{J, $7\times8$}{\includegraphics[scale=\scale]{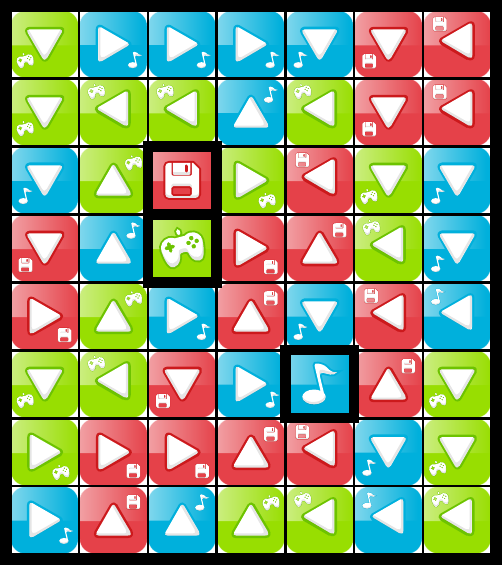}}
    \smallskip\\
    \subcaptionbox{J, $7\times8$}{\includegraphics[scale=\scale]{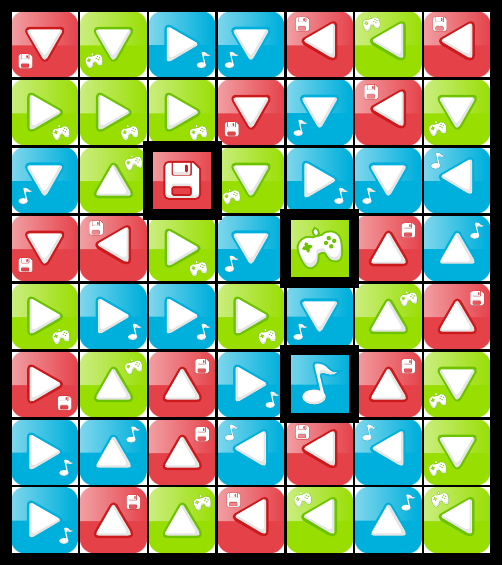}} &
    \subcaptionbox{J, $7\times8$}{\includegraphics[scale=\scale]{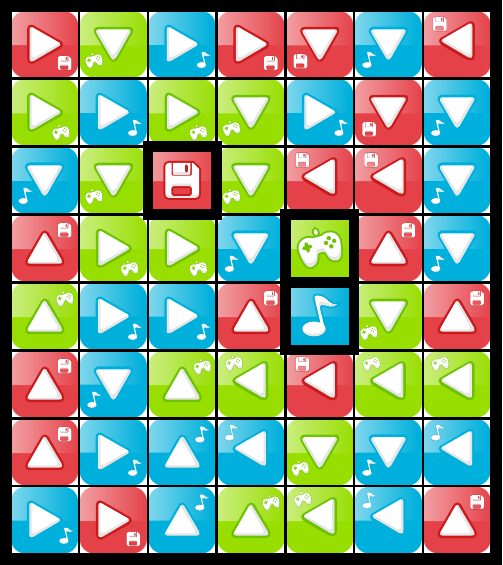}} &
    \subcaptionbox{J, $7\times8$}{\includegraphics[scale=\scale]{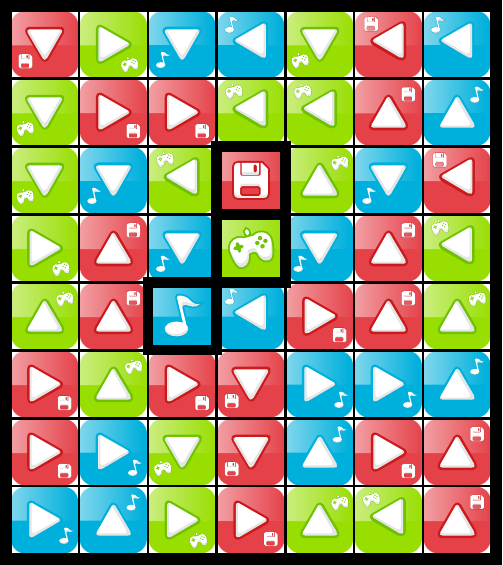}}
    \smallskip\\
    \subcaptionbox{J, $7\times8$}{\includegraphics[scale=\scale]{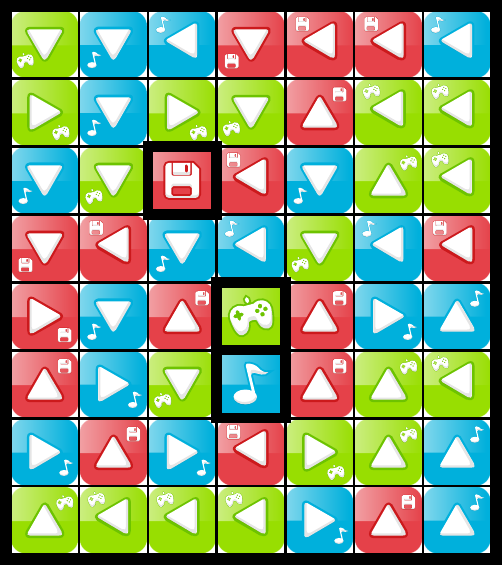}} &
    \subcaptionbox{J, $7\times8$}{\includegraphics[scale=\scale]{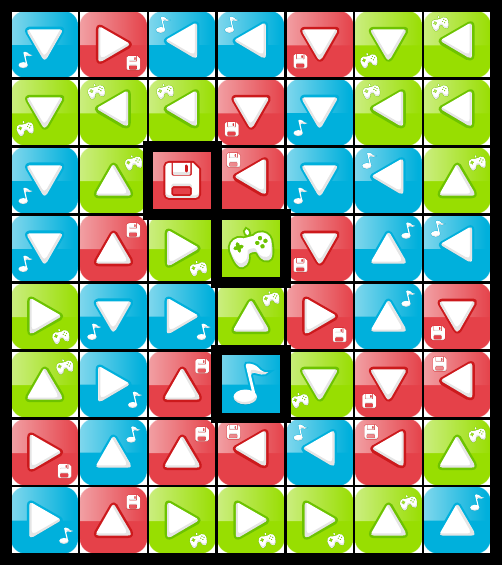}} &
    \subcaptionbox{J, $7\times8$}{\includegraphics[scale=\scale]{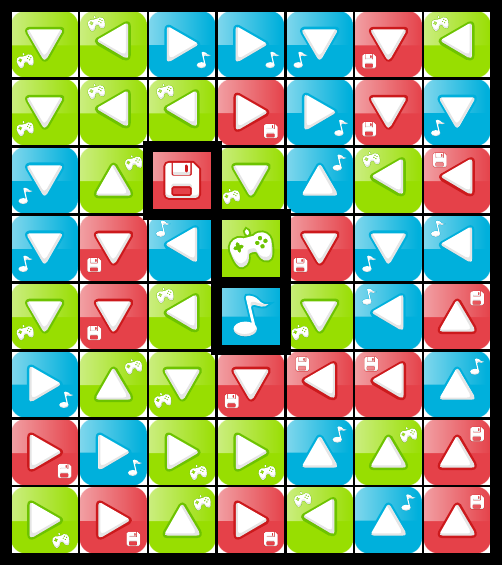}}
  \end{tabular}
  \caption{Solutions from the automated solver for sinks in the J pattern
    of size $7\times8$.}
  \label{fig:solver-solutions-J-7x8}
\end{figure}

\begin{figure}
  \centering
  \tabcolsep=0.75em
  \begin{tabular}{ccc}
    \subcaptionbox{J, $8\times7$}{\includegraphics[scale=\scale]{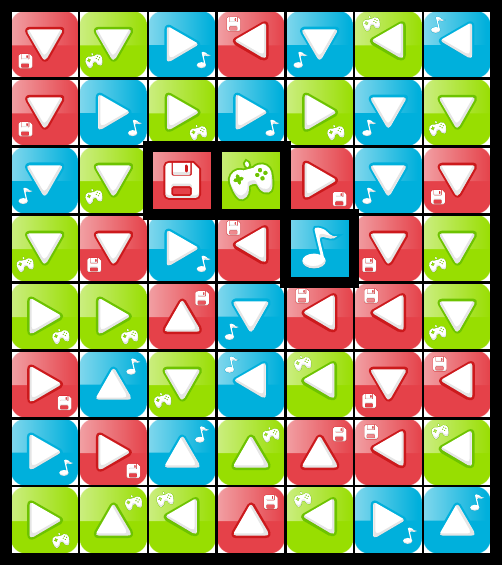}} &
    \subcaptionbox{J, $8\times7$}{\includegraphics[scale=\scale]{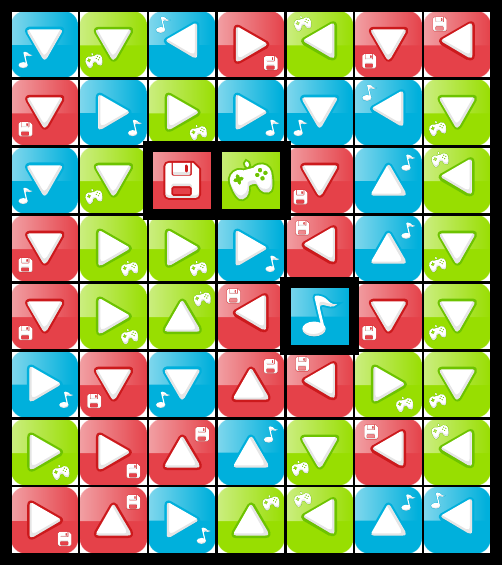}} &
    \subcaptionbox{J, $8\times7$}{\includegraphics[scale=\scale]{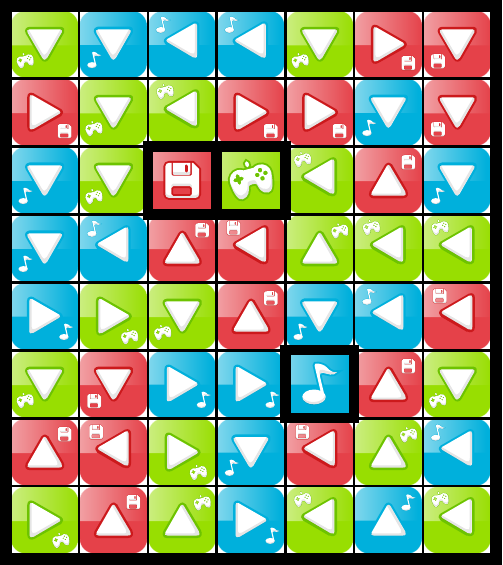}}
    \smallskip\\
    \subcaptionbox{J, $8\times7$}{\includegraphics[scale=\scale]{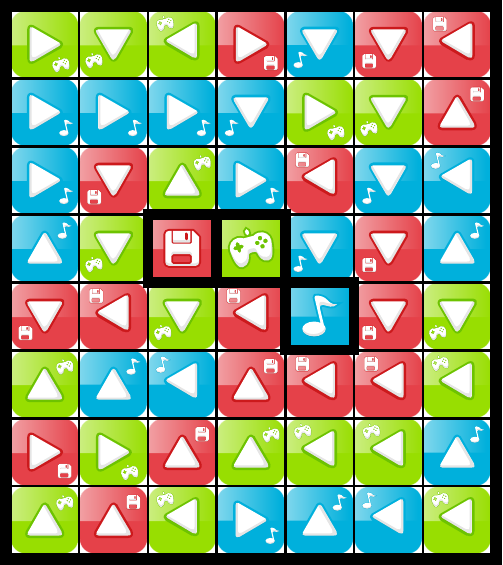}} &
    \subcaptionbox{J, $8\times7$}{\includegraphics[scale=\scale]{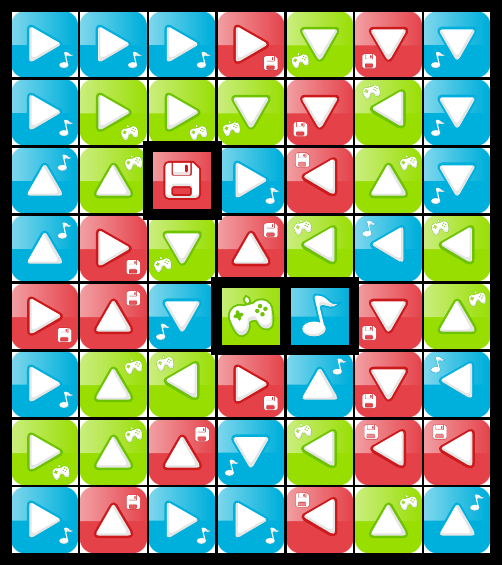}} &
    \subcaptionbox{J, $8\times7$}{\includegraphics[scale=\scale]{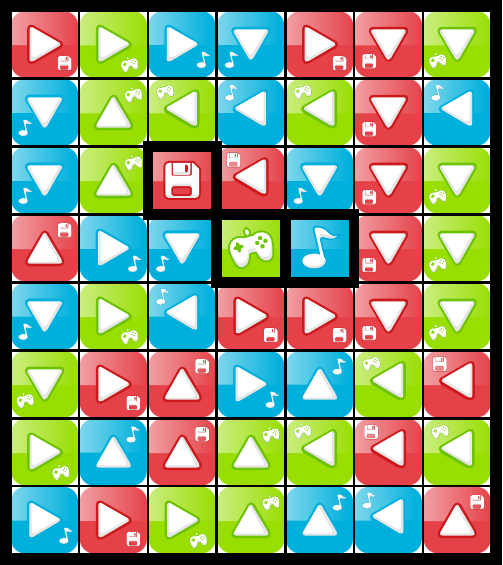}}
  \end{tabular}
  \caption{Solutions from the automated solver for sinks in the J pattern
    of size $8\times7$.}
  \label{fig:solver-solutions-J-8x7}
\end{figure}

\begin{figure}
  \centering
  \begin{tabular}{ccc}
    \subcaptionbox{L, $6\times9$}{\includegraphics[scale=\scale]{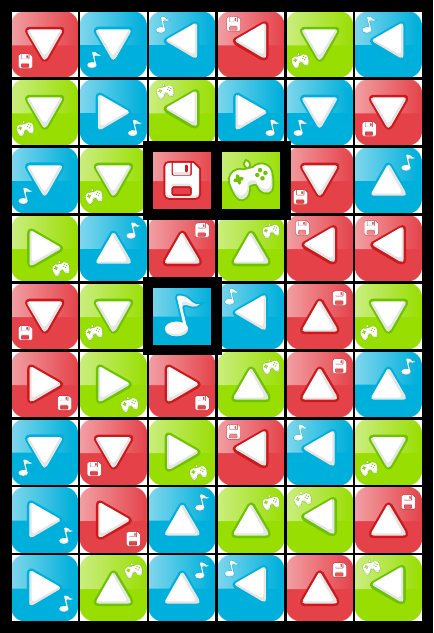}} &
    \subcaptionbox{L, $6\times9$}{\includegraphics[scale=\scale]{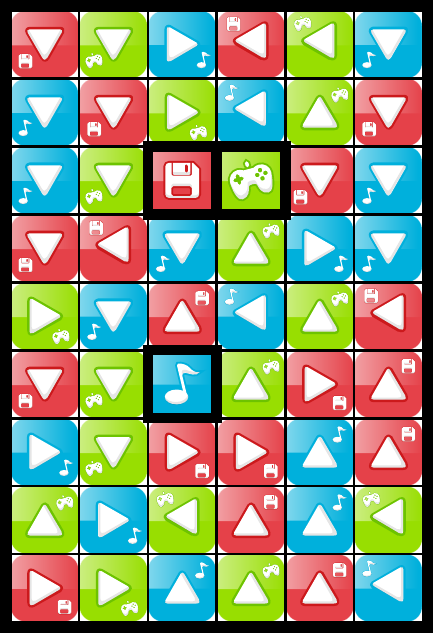}} &
    \subcaptionbox{L, $6\times9$}{\includegraphics[scale=\scale]{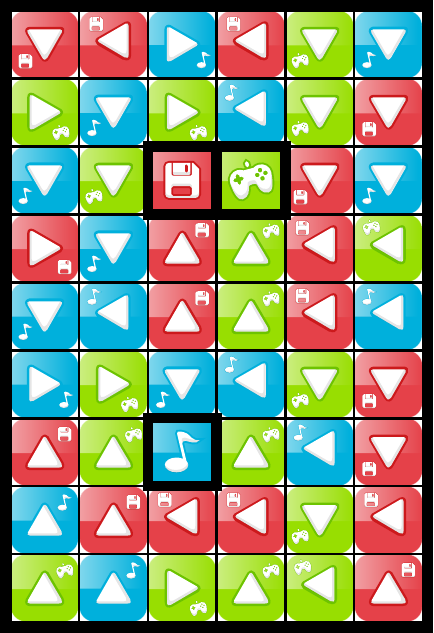}}
    \smallskip\\
    \subcaptionbox{L, $6\times9$}{\includegraphics[scale=\scale]{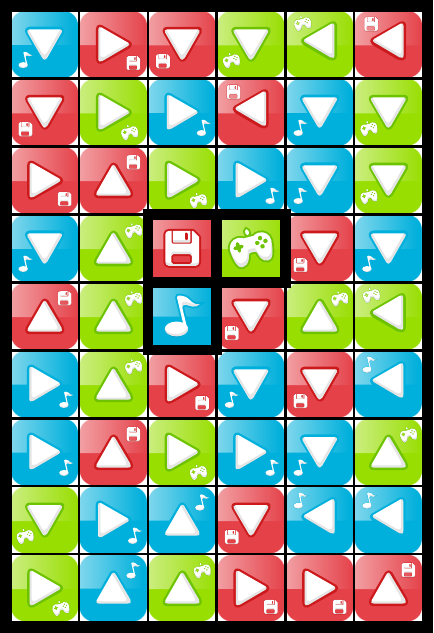}} &
    \subcaptionbox{L, $6\times9$}{\includegraphics[scale=\scale]{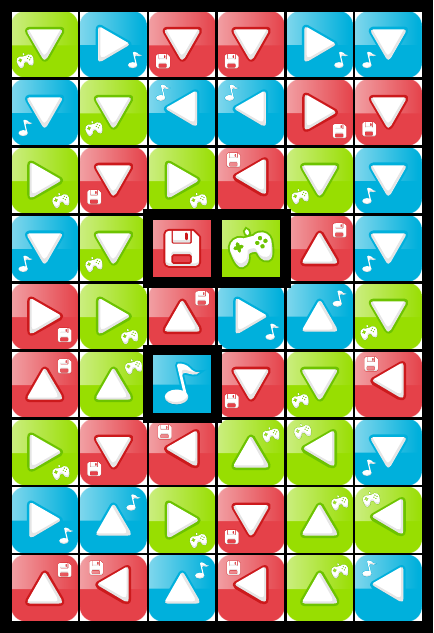}} &
    \subcaptionbox{L, $6\times9$}{\includegraphics[scale=\scale]{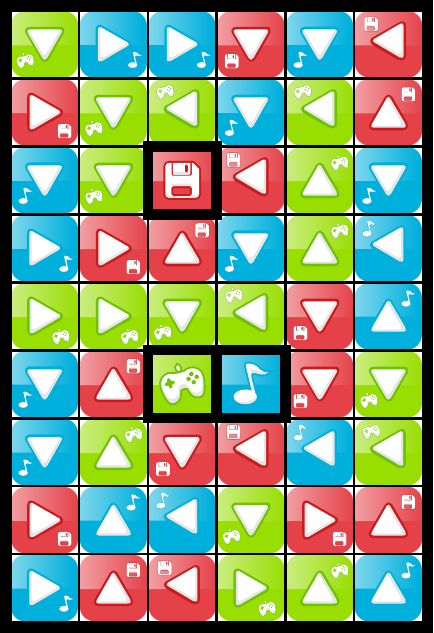}}
    \smallskip\\
    \subcaptionbox{L, $6\times9$}{\includegraphics[scale=\scale]{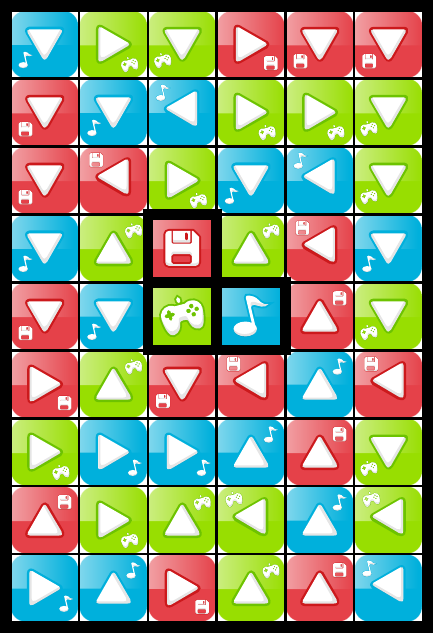}} &
    \subcaptionbox{L, $6\times9$}{\includegraphics[scale=\scale]{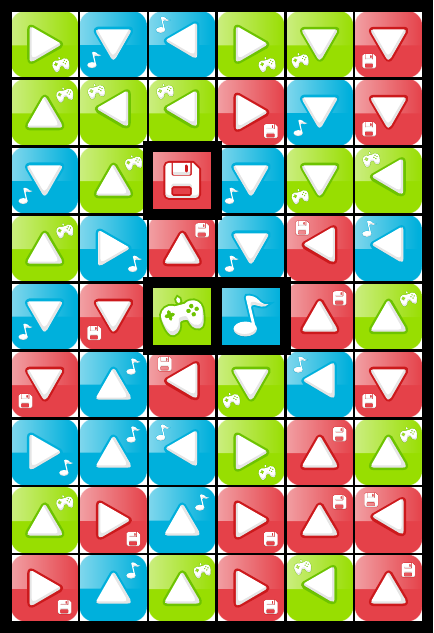}} &
    \subcaptionbox{L, $6\times9$}{\includegraphics[scale=\scale]{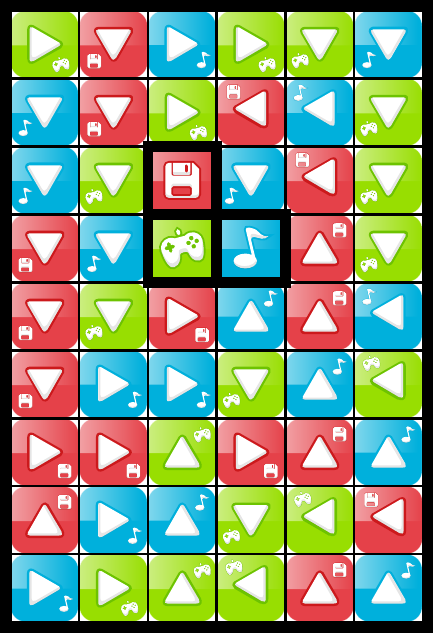}}
  \end{tabular}
  \caption{Solutions from the automated solver for sinks in the L pattern
    of size $6\times9$.}
  \label{fig:solver-solutions-L-6x9}
\end{figure}

\begin{figure}
  \centering
  \begin{tabular}{ccc}
    \subcaptionbox{L, $7\times7$}{\includegraphics[scale=\scale]{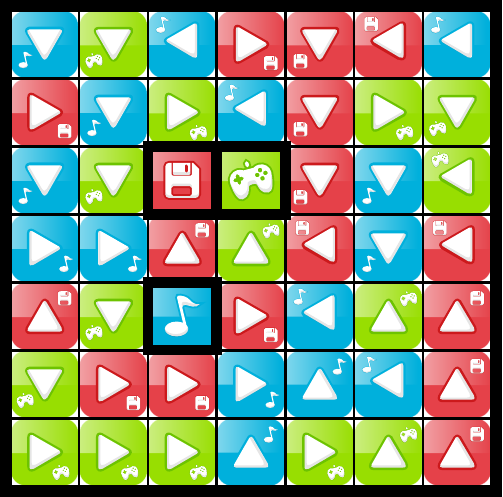}} &
    \subcaptionbox{L, $7\times7$}{\includegraphics[scale=\scale]{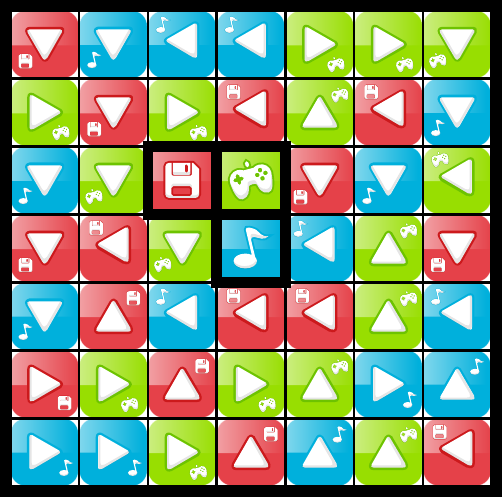}} &
    \subcaptionbox{L, $7\times7$}{\includegraphics[scale=\scale]{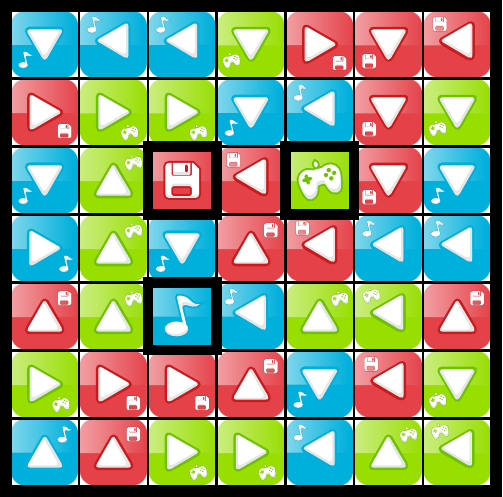}}
    \smallskip\\
    \subcaptionbox{L, $7\times7$}{\includegraphics[scale=\scale]{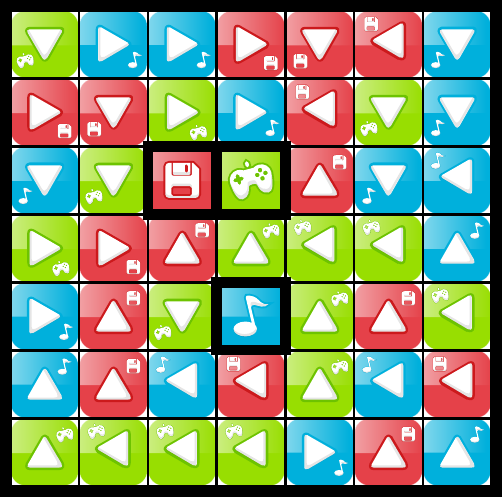}} &
    \subcaptionbox{L, $7\times7$}{\includegraphics[scale=\scale]{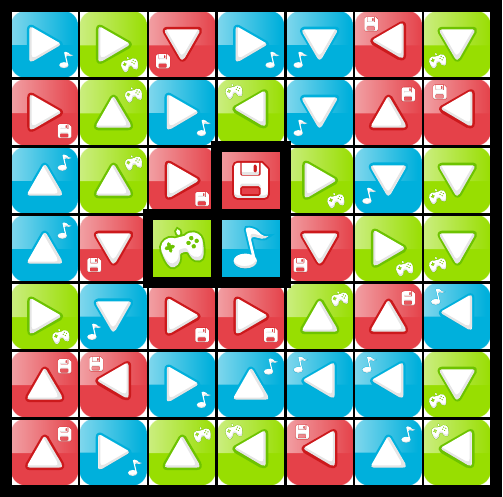}}
  \end{tabular}
  \caption{Solutions from the automated solver for sinks in the L pattern
    of size $7\times7$.}
  \label{fig:solver-solutions-L-7x7}
\end{figure}

\begin{figure}
  \centering
  \begin{tabular}{cc}
    \subcaptionbox{Y, $7\times8$}{\includegraphics[scale=\scale]{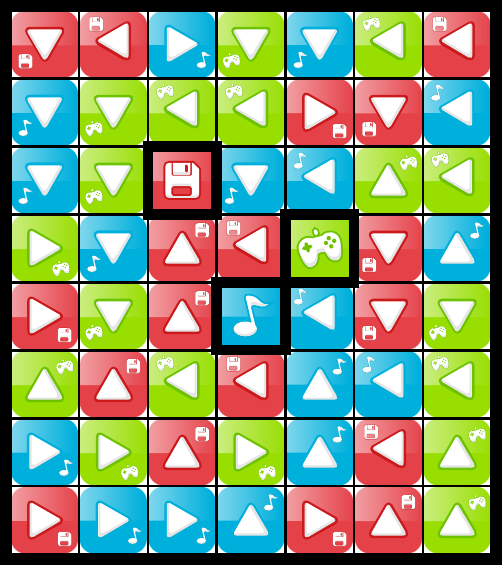}} &
    \subcaptionbox{Y, $7\times8$}{\includegraphics[scale=\scale]{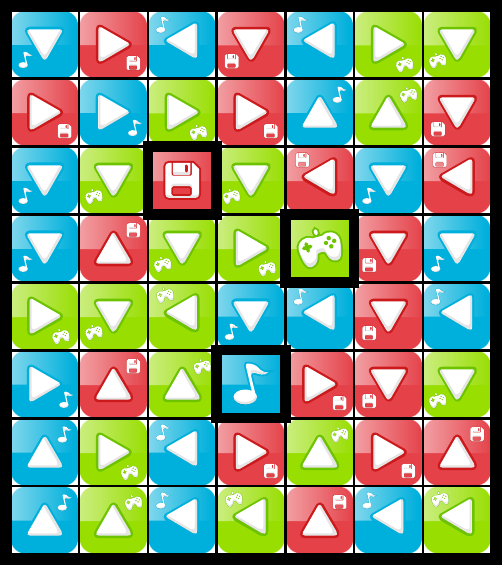}}
    \smallskip\\
    \subcaptionbox{Y, $7\times8$}{\includegraphics[scale=\scale]{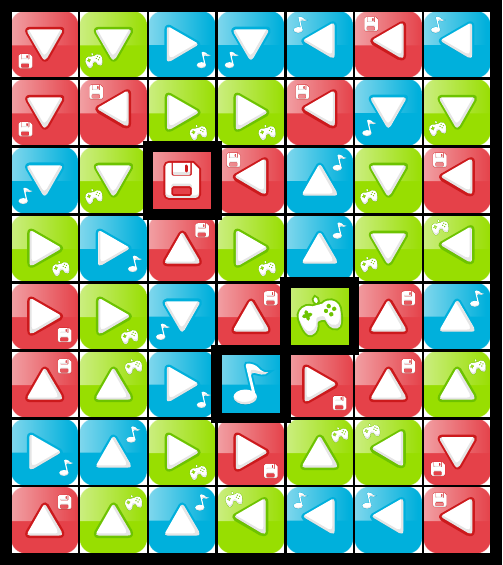}} &
    \subcaptionbox{Y, $7\times8$}{\includegraphics[scale=\scale]{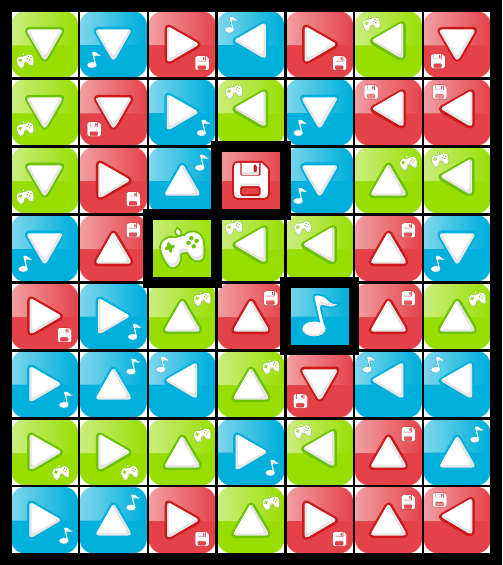}}
  \end{tabular}
  \caption{Solutions from the automated solver for sinks in the Y pattern
    of size $7\times8$.}
  \label{fig:solver-solutions-Y-7x8}
\end{figure}

\begin{figure}
  \centering
  \begin{tabular}{cc}
    \subcaptionbox{/, $7\times8$}{\includegraphics[scale=\scale]{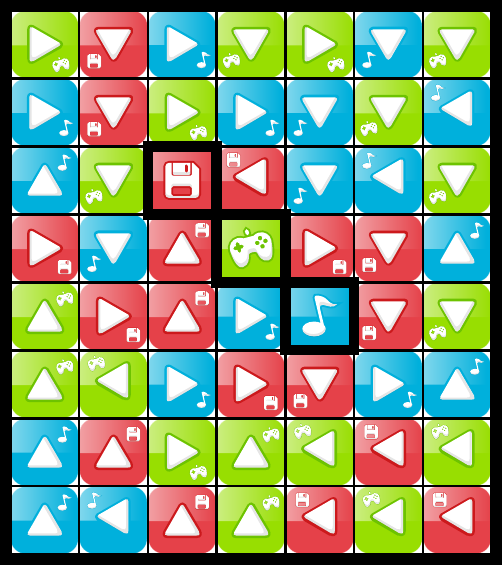}} &
    \subcaptionbox{/, $7\times8$}{\includegraphics[scale=\scale]{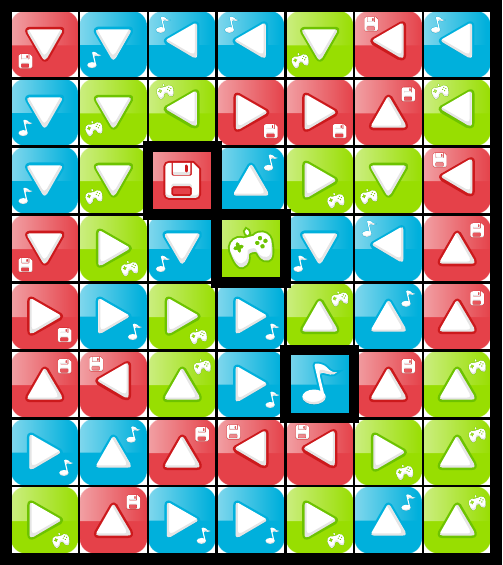}}
  \end{tabular}
  \caption{Solutions from the automated solver for sinks in the / pattern
    of size $7\times8$.}
  \label{fig:solver-solutions-/-7x8}
\end{figure}

\end{document}